\documentclass[11pt]{article}
\usepackage{pset}
\usepackage{breqn}
\usepackage[normalem]{ulem}
\usepackage{mathtools}

\mathtoolsset{showonlyrefs}

\usepackage{constants}
\newcommand{\ones}{\vec{1}}
\newcommand{\calD}{\mathcal{D}}
\newcommand{\sos}[1]{\vdash_{#1}}
\renewcommand{\E}{\mathbb{E}}
\newcommand{\mul}{\text{Mul}}
\newcommand{\Bin}{\text{Bin}}
\newcommand{\vp}{\vec{p}}
\newcommand{\vq}{\vec{q}}
\newcommand{\eqd}{\overset{d}{=}}

\newcommand{\Ber}{\text{Ber}}
\newcommand{\psE}{\tilde{\mathbb{E}}}
\newcommand{\calA}{\mathcal{A}}
\newcommand{\calP}{\mathcal{P}}
\newcommand{\calJ}{\mathcal{J}}
\newcommand{\calV}{\mathcal{V}}
\newcommand{\calK}{\mathcal{K}}
\newcommand{\calB}{\mathcal{B}}
\newcommand{\mother}{\text{mother}}
\newcommand{\father}{\text{father}}

\renewcommand{\P}{\mathbb{P}}
\newcommand{\bone}{{\mathds{1}}} 
\renewcommand{\A}[1]{\calA_{#1}}

\newconstantfamily{?}{symbol=???}

\newcommand\numeq[1]%
  {\stackrel{\scriptscriptstyle(\mkern-1.5mu#1\mkern-1.5mu)}{=}}
\newcommand\numle[1]%
  {\stackrel{\scriptscriptstyle(\mkern-1.5mu#1\mkern-1.5mu)}{\le}}

\title{Efficiently Learning Structured Distributions from Untrusted Batches}
\author{Sitan Chen\thanks{EECS, Massachusetts Institute of Technology. Email: {\tt sitanc@mit.edu}. This work was supported in part by a Paul and Daisy Soros Fellowship, NSF CAREER Award CCF-1453261, and NSF Large CCF-1565235. This work was done in part while S.C. was an intern at Microsoft Research AI.} \and Jerry Li\thanks{Microsoft Research AI. Email: {\tt jerrl@microsoft.com}.} \and Ankur Moitra\thanks{Department of Mathematics, Massachusetts Institute of Technology. Email: {\tt moitra@mit.edu}. This work was supported in part by a Microsoft Trustworthy AI Grant, NSF CAREER Award CCF-1453261, NSF Large CCF-1565235, a David and Lucile Packard Fellowship, an Alfred P. Sloan Fellowship and an ONR Young Investigator Award.}}

\newtheorem{axioms}{Axioms}

\newtheorem{Alg}{Algorithm}
\newcommand{\myalg}[3]{
\medskip
\footnotesize{
\fbox{
\parbox{5.5in}{
\begin{Alg}\label{#1}{\sc #2}\\
{\tt #3}
\end{Alg}
}}
\medskip
}}

\usepackage{hyperref}
\hypersetup{colorlinks=true,citecolor=red,linkcolor=blue}

\begin{document}
\maketitle

\begin{abstract}
We study the problem, introduced by Qiao and Valiant \cite{qiao2017learning}, of learning from untrusted batches. 
Here, we assume $m$ users, all of whom have samples from some underlying distribution $\vp$ over $1, \ldots, n$.
Each user sends a batch of $k$ i.i.d. samples from this distribution; however an $\epsilon$-fraction of users are untrustworthy and can send adversarially chosen responses.
The goal of the algorithm is then to learn $\vp$ in total variation distance.
When $k = 1$ this is the standard robust univariate density estimation setting and it is well-understood that $\Omega (\epsilon)$ error is unavoidable.
Suprisingly, \cite{qiao2017learning} gave an estimator which improves upon this rate when $k$ is large.
Unfortunately, their algorithms run in time which is exponential in either $n$ or $k$.

We first give a sequence of polynomial time algorithms whose estimation error approaches the information-theoretically optimal bound for this problem.
Our approach is based on recent algorithms derived from the sum-of-squares hierarchy, in the context of high-dimensional robust estimation.
We show that algorithms for learning from untrusted batches can also be cast in this framework, but by working with a more complicated set of test functions. 

It turns out that this abstraction is quite powerful, and can be generalized to incorporate additional problem specific constraints.
Our second and main result is to show that this technology can be leveraged to build in prior knowledge about the shape of the distribution.
Crucially, this allows us to reduce the sample complexity of learning from untrusted batches to polylogarithmic in $n$ for most natural classes of distributions, which is important in many applications.
To do so, we demonstrate that these sum-of-squares algorithms for robust mean estimation can be made to handle complex combinatorial constraints (e.g. those arising from VC theory), which may be of independent technical interest.
\end{abstract}

	\newpage
	

\section{Introduction}

Qiao and Valiant \cite{qiao2017learning} introduced the following basic problem, in robust distribution learning, that they called {\em learning with untrusted batches}:

\begin{enumerate}

\item[(a)] We are given $m$ batches, consisting of $k$ samples each. Furthermore the samples come from a discrete domain of size $n$. Each uncorrupted batch has the property that its samples were drawn i.i.d. from some distribution $p_i$ that is $\eta$-close in total variation distance\footnote{The total variation distance between distributions $p,q$ over a domain $D$ is defined to be $\max_{S\subseteq D} p(S) -q(S)$} to a distribution $p$ that is common to all the batches. Moreover a $1-\epsilon$ fraction of the batches are uncorrupted. 

\item[(b)] The remaining $\epsilon$ fraction of the batches are arbitrarily corrupted. In fact, an adversary is allowed to choose the contents of the corrupted batches after observing all of the uncorrupted batches.

\end{enumerate}

\noindent The basic question is: {\em How well can we estimate $p$ in total variation distance?} The key features of this problem are designed to model some of the main challenges in federated learning. In particular, we get batches of data from different users, but no batch is large enough by itself to learn an accurate model. In fact, the batches are generated from heterogenous sources because the ideal model for one user is often different than the ideal model for another. Additionally some of the batches are arbitrarily corrupted by an adversary who wishes to game our learning algorithm. In many applications, a non-trivial fraction of the data is supplied by malicious users. The meta question is: {\em Can we leverage information across the batches to learn an accurate model?}

In fact, the setup of learning with untrusted batches seems to model many other scenarios of interest. Our main focus will be settings where we have some additional structure or prior knowledge about the distributions we would like to learn. For example, suppose we want to estimate the demand curve across heterogenous groups. In particular, let $q_1 < q_2 < \cdots < q_n$ be a collection of increasing prices. Then set $p_{i,j}$ to be the probability that a random individual from group $i$ would buy the product when offered a price $q_j$ but not at the price $q_{j+1}$. We may not have enough data from each group to accurately estimate $p_i$. Nevertheless we can hope to leverage data across the groups to estimate an aggregate curve $p$ that is a good approximation to each $p_i$. Interestingly, the goal of being robust to an $\epsilon$-fraction of the batches being corrupted now takes on a different meaning in this setting: We are asking whether we can estimate $p$ from data collected across the various groups in such a way that no  $\epsilon$-fraction of the groups can bias our estimates too much. 


Qiao and Valiant \cite{qiao2017learning} showed that it is possible to estimate $p$ within $$O\left(\frac{\epsilon}{\sqrt{k}} + \eta\right)$$ in total variation distance, from untrusted batches. Moreover they showed that this is the best possible up to constant factors. The somewhat surprising aspect of their bound is that it improves with larger $k$. This is a consequence of the ``tensorization" property of the total variation distance which roughly says that the total variation distance between two distributions grows by at least a $\Omega(\sqrt{k})$ factor when we take $k$ repetitions. 

However, Qiao and Valiant \cite{qiao2017learning} were only able to give an exponential time algorithm. Their approach was to estimate $p$ by estimating the total probability it assigns to every subset of the domain. Each of these subproblems is again a problem of learning with untrusted batches, but one on a discrete domain with just two elements. Qiao and Valiant \cite{qiao2017learning} gave another algorithm, but one that requires $\eta = 0$ \--- i.e. each of the uncorrupted batches must be generated from the same underlying distribution. Their second algorithm was based on low-rank tensor approximation. They wrote down an order $k$ tensor whose entries represent the probability of seeing any particular $k$ tuple of samples as a batch, and showed that some slice of this tensor is an accurate estimate of $p$. This algorithm also has the drawback that in order to estimate the entries of the tensor, you need $n^k$ samples. In most applications, it would be infeasible to have so much data that you see essentially every possible batch. Their work left open the problem of getting efficient algorithms for learning with untrusted batches.

\subsection{Our Results}

In this work, we use the sum-of-squares hierarchy to design new algorithms for the problem of learning from untrusted batches. An important feature of our approach is that it is easy to incorporate additional prior information about the shape of the distribution, and get even better running time and sample complexity. But first, as a warm up, we will study the original learning with untrusted batches problem. We give a sequence of polynomial time algorithms whose estimation error approaches the information-theoretically optimal bound: 

\begin{thm}[See Theorem~\ref{thm:main_basic} for formal statement]
Fix any integer $t \geq 4$. There is a polynomial time algorithm to estimate $p$ to within
$$O\left(\frac{\epsilon^{1-1/t}}{\sqrt{k/t}} + \eta\right)$$ in total variation distance from $m$ $\epsilon$-corrupted batches, each of size $k$. Moreover the number of batches we need is polynomial in $n$. 
\end{thm}

\noindent This result improves over the $2^n$ time algorithm of Qiao and Valiant \cite{qiao2017learning}. Note that the other algorithm of Qiao and Valiant \cite{qiao2017learning} runs in time $n^k$ but only works in the special case where $\eta = 0$ \--- i.e. all the uncorrupted batches come from the same underlying distribution. Moreover, in the above result, if we set $t = \log 1/\epsilon$ then we get within a polylogarithmic factor of the optimal estimation error, but at the expense of running in quasipolynomial time:

\begin{cor}
There is an algorithm to estimate $p$ to within $$O\left(\frac{\epsilon \sqrt{\log 1/\epsilon} }{\sqrt{k}} + \eta\right)$$ in total variation distance from $m$ $\epsilon$-corrupted batches, each of size $k$. Moreover the running time and the number of batches we need are polynomial in $n^{\log 1/\epsilon}$. 
\end{cor}

Finally, we come to what we believe to be our main contribution. In many applications, getting samples is expensive and we might only be able to afford a number of samples that is sublinear in the size of the domain. In such cases, it is important to utilize additional information such as prior knowledge about the shape of the distribution. Indeed, this is the case in the example we discussed earlier, where we often know that the distribution $p$ satisfies the monotone hazard rate condition. It is known that such distributions can be well-approximated by piecewise polynomial functions \cite{chan2013learning,chan2014efficient,acharya2017sample}. 

In fact, the idea of imposing structure on the underlying distribution has a long and storied history in statistics and machine learning where it leads to better estimation rates and algorithms that use fewer samples \cite{brunk1955maximum, hildreth1954point, wegman1970maximum}. We ask: {\em Can prior information about the shape of a distribution be leveraged to get better algorithms for learning from untrusted batches?} Our main result is:

\begin{thm}[See Theorem~\ref{thm:main_shape} for formal statement]
Fix any integer $t \geq 4$. If $p$ is approximated by an $s$-part piecewise polynomial function with degree at most $d$, there is a polynomial time algorithm to estimate $p$ to within
$$O\left(\frac{\epsilon^{1-1/t}}{\sqrt{k/t}} + \eta\right)$$ in total variation distance from $m$ $\epsilon$-corrupted batches, each of size $k$. Moreover the number of batches we need is polylogarithmic in $n$ and polynomial in $s$ and $d$. 
\end{thm}

\noindent 
While the problem of learning a piecewise polynomial distribution may not seem natural in applications, previous work of~\cite{chan2013learning,chan2014efficient,acharya2017sample} has demonstrated that this can be combined with results from approximation theory~\cite{timan2014theory} to achieve strong density estimation results for a large class of distribution families such as log-concave distributions, Gaussians, monotone distributions, monotone hazard rate distributions, Binomial distributions, Poisson distributions, and mixtures thereof~\cite{acharya2017sample}.

In the next subsection, we describe our main techniques at a high level. The main takeaway is that the sum-of-squares hierarchy gives a seamless way to incorporate prior information about the structure into the estimation problem, which can lead to much better algorithms (in our case we are able to get sublinear sample complexity).

 
 \subsection{Our Techniques}
 
 Recently, there has been a flurry of progress in high-dimensional robust estimation \cite{diakonikolas2019robust, lai2016agnostic, charikar2017learning, diakonikolas2017being}. While the techniques seem to be quite different from each other \--- some relying on iterative filtering algorithms to remove outliers, and others relying on sum-of-squares proofs of identifiability \--- at their heart, they are about finding ways to re-weight the empirical distribution on the observed samples in such a way that it has bounded moments along any one-dimensional projection \cite{hopkins2018mixture, kothari2018robust, diakonikolas2018list}. 
 
Our main observation is that algorithms for learning from untrusted batches can also be derived from this framework, but by working with a different family of test functions. When we consider moments of a one-dimensional projection, we are looking at test functions that are unit vectors (or tensor powers of them) in the $\ell_2$-norm. In comparison, the exponential time algorithm of Qiao and Valiant \cite{qiao2017learning} tries all ways of partitioning the domain into two sets. We can equivalently think about it as choosing a test vector (or tensor power of one) that has unit $\ell_\infty$-norm. In this way, we study the families of distributions for which we can find a sum-of-squares certificate that they have bounded moments with respect to unit $\ell_\infty$ test functions. We show that the multinomial distribution has this property, and using the proofs-to-algorithms methodology \cite{hopkins2018mixture, kothari2018robust}, this gives our improved algorithm for the general problem of learning with untrusted batches. 

The beauty of this common abstraction is that it flexibly allows us to build in other problem specific constraints, like shape constraints on $p$. 
Here, classical results from VC theory~\cite{vapnik1974theory,devroye2001combinatorial} say that it suffices to learn the distribution in a weaker norm (see Definition~\ref{def:ak}) than total variation distance, which has fewer degrees of freedom.
From our perspective, the change is that, in this case, instead of allowing all unit $\ell_\infty$ test functions, we only have to consider those which come from tensor powers of a vector that has a bounded number of sign changes. 
However, encoding this constraint in the sum-of-squares hierarchy is quite non-trivial, as it is not clear how to encode this combinatorial constraint within the algebraic language of the sum-of-squares proof system.
To get around this, we demonstrate that we can relax the combinatorial constraint into a linear algebraic one, namely, sparsity in the \emph{Haar wavelet basis}. 
We then exploit properties of the Haar wavelet basis to encode this constraint into our relaxation. 
{\em The main open question of our work is to push this philosophy further, and explore what other sorts of provably robust algorithms can be built out of different choices of test functions. }

\subsection{Related work}
The problem of learning from untrusted batches was introduced by~\cite{qiao2017learning}, and is motivated by problems in reliable distributed learning such as \emph{federated learning}~\cite{44822,konevcny2016federated}.
In the TCS community, the problem of learning from batches has been considered in a number of settings~\cite{levi2013testing,tian2017learning}, but these results cannot tolerate noise in the data.

More generally, the question of univariate density estimation, and specifically, density estimation of structured distributions, has a vast literature and we cannot hope to fully survey it here.
See~\cite{barlow1972statistical} for a survey of classical results in the area.
Many different natural structural assumptions have been considered in the statistics and learning theory communities, such as monotonicity~\cite{Grenander:56,Groeneboom:85,Birge:87,Birge:87b,JW:09}, monotone hazard rate~\cite{chan2013learning,cole2014sample,huang2018making}, unimodality~\cite{PrakasaRao:69,Wegman:70,Fougeres:97}, convexity and concavity~\cite{HansonP:76,KoenkerM:10aos}, log-concavity~\cite{BRW:09aos,DumbgenRufibach:09,Walther09}, $k$-modality~\cite{ChanTong:04,BW07aos,GW09sc,BW10sn}, smoothness~\cite{Brunk:58,KerkPic92, Don95, KPicT96, Donoho96, Donoho98}, and mixtures of structured distributions~\cite{RednerWalker:84,TSM:85,Lindsay:95,Dasgupta:99,DasguptaSchulman:00,AroraKannan:01, VempalaWang:02,FOS:05focs, AchlioptasMcSherry:05,KMV:10,MoitraValiant:10,DDS12stoc,DDS12soda,DDOST13focs,diakonikolas2016fourier,daskalakis2016size,diakonikolas2016optimal,diakonikolas2016properly}.
The reader is referred to~\cite{o2016nonparametric,diakonikolas2016learning} for a more extensive review of this vast literature.
Recently it has been demonstrated that the classical piecewise polynomial (or spline) methods, see e.g.~\cite{WegW83, Stone94, Stone97,WillettN07}, can be adapted to obtain general estimators for almost all of these problems with nearly-optimal sample complexity and runtime~\cite{chan2013learning,chan2014efficient,CDSS14b,ADHLS15,acharya2017sample}.
While these estimators are typically tolerant of worst-case noise, it is unclear how to adapt them to the batch setting, to obtain improved statistical rates.

Finally, our work is also related to a recent line of work on robust statistics~\cite{diakonikolas2019robust, lai2016agnostic, charikar2017learning, diakonikolas2017being,hopkins2018mixture, kothari2018robust}, a classical problem dating back to the 60s and 70s~\cite{anscombe1960rejection,tukey1960survey,huber1992robust,tukey1975mathematics}.
See~\cite{li2018principled,steinhardt2018robust} for a more comprehensive survey of this line of work.
We remark that the majority of this work focuses on estimation in $\ell_2$-norm or Frobenius norm, with two notable exceptions:~\cite{balakrishnan2017computationally} uses learning in a sparsity-inducing norm to improve the sample complexity for sparse mean estimation, and~\cite{steinhardt2018resilience} gives an information-theoretic characterization of when mean estimation in general norms is possible, but they do not give efficient algorithms.
Our techniques are most closely related to the sum-of-squares based algorithms of~\cite{hopkins2018mixture, kothari2018robust}, and this general technique has also found application in other robust learning problems such as robust regression~\cite{klivans2018efficient} and list-decodable regression~\cite{karmalkar2019list,raghavendra2019list}.

\subsection{Organization}

In Section~\ref{sec:techniques}, we provide a high-level overview of our techniques. In Section~\ref{sec:technical_prelims}, we give notation, a formal description of the generative model, a recap of the key SoS tools needed, and show a sum-of-squares proof that multinomial distributions have bounded moments. In Section~\ref{sec:batches} we give a proof of Theorem~\ref{thm:main_basic}. In Section~\ref{sec:shape}, we give a proof of Theorem~\ref{thm:main_shape}. The technical heart of this work is Section~\ref{sec:matrixsos}, where we fill in the details on how to efficiently encode key constraints from our SoS relaxations using matrix SoS. In Appendix~\ref{app:defer}, we provide proofs deferred from earlier sections.




\section{High-Level Argument}
\label{sec:techniques}

In this section we give an overview of how we prove Theorems~\ref{thm:main_basic} and \ref{thm:main_shape}. The ideas required for the latter are a strict subset of those for the former, so we first describe the aspects common to both proofs before elaborating in Section~\ref{subsec:AK} and \ref{subsec:quantifyVnK} on techniques specific to Theorem~\ref{thm:main_shape}, which we view as the main contribution of this work. As these latter sections are somewhat technical, readers new to the use of sum-of-squares for robust mean estimation may feel free to skip them on first reading, as the other sections will be sufficient for understanding the proof of Theorem~\ref{thm:main_basic}.

\subsection{Robust Mean Estimation}

We first recast the problem of learning from untrusted batches as a generalization of the problem of robustly estimating the mean of a multinomial distribution in $L_1$ distance.

To the $i$-th batch of $k$ samples $Z^i = (Z^i_1,...,Z^i_k)$ from $[n]$ we may associate the vector of frequencies $Y_i\in\Delta^n$ (where $\Delta^n\subset\R^n$ is the probability simplex) given by \begin{equation}
	(Y_i)_j = \frac{1}{k}\sum^k_{\nu = 1}\bone[Z^i_{\nu} = j] \ \forall j\in[n].
\end{equation} If $Z^1,...,Z^N$ are independent batches of $k$ iid draws from $\vp_1,...,\vp_N$ respectively, then $Y_1,...,Y_N$ are independent draws from $\mul_k(\vp_1),...,\mul_k(\vp_N)$ respectively, where $\mul_k(\vp_i)$ is defined to be the normalized multinomial distribution given by $k$ draws from $\vp_i$. We can think of the learning algorithm as taking in vectors $X_1,...,X_N\in\Delta^n$, such that a $(1-\epsilon)N$-sized subset of them, indexed by $S_g\subset[N]$, are independent draws from $\mul_k(\vp_j)$ for $j\in S_g$, and the remaining points are arbitrary vectors in $\Delta^n$. The goal of the learning algorithm is to learn $\vp$ in $L_1$ distance. Note that when $\vp_i = \vp$ for all $i\in S_g$, this is precisely the problem of robustly estimating the mean $\vp$ of a (normalized) multinomial distribution.

For simplicity, we will assume that $\delta = 0$ for the rest of this subsection, i.e. that $\vp_1 = \cdots = \vp_N$. Indeed, one appealing feature of our techniques is the ease with which one can extend the techniques we describe below to handle the case of nonzero $\delta$.

\subsection{Searching for a Moment-Bounded Subset}
\label{subsec:search}

A recurring theme in the robust learning literature \cite{diakonikolas2019robust, lai2016agnostic, hopkins2018mixture, kothari2018robust, diakonikolas2018list} is that one can detect corruptions in the data by looking for anomalies in the empirical moments. In our setting, one useful feature of multinomial distributions $\mul_k(\vp)$ is that their moments up to degree $k$ satisfy sub-Gaussian-type bounds.

\begin{thm}[\cite{latala1997estimation}]
 	For a (normalized) binomial random variable $Z\sim\frac{1}{k}\cdot \text{Bin}(k,p)$, \begin{equation}\E[(Z - p)^t]^{1/t} \lesssim \sqrt{t/k}\end{equation} for any even $t\le k$.\label{thm:latala}
\end{thm}

Multinomial distributions inherit these same properties:

\begin{lem}
	For any discrete distribution $\vp$ and any vector $v\in\{\pm 1\}^n$, if $X\sim\mul_k(\vp)$, then \begin{equation}
		\E[\langle X - \vp, v\rangle^t]^{1/t} \lesssim \sqrt{t/k}
	\end{equation} for any even $t\le k$.
\end{lem}

At a high level, our algorithms will search for a $(1-\epsilon)N$-sized subset $S$ of the samples whose empirical moments satisfy these bounds, namely \begin{equation}
	\frac{1}{|S|}\sum_{i\in S}\langle X_i - \hat{\vp}, v\rangle^t \le (8t/k)^{t/2} \ \ \ \forall v\in\{\pm 1\}^n, \label{eq:want}
\end{equation} where $\hat{\vp} = \frac{1}{|S|}\sum_{i\in S}X_i$ is the empirical mean of $S$. This search problem can be reformulated as solving some system $\calP$ of polynomial equalities and inequalities (see Section~\ref{sec:batches} for a formal specification). So if we could solve this system and argue that the empirical mean of \emph{any} subset $S\subset[N]$ which satisfies the system is $O(\epsilon/\sqrt{k})$-close in $L_1$ to $\vp$, then we'd be done.

There are two complications to this approach: \begin{enumerate}[label=(\Alph*)]
	\item The problem of solving polynomial systems is $\NP$-hard in general.\label{issue:polyhard}
	\item Constraint \eqref{eq:want} is a collection of exponentially many constraints.\label{issue:constraints}
\end{enumerate}

By now it is well-understood how to circumvent issues like \ref{issue:polyhard}: use the sum-of-squares (SoS) hierarchy to relax the problem of searching for a single solution to $\calP$, or even a \emph{distribution} over solutions, to the problem of searching for a \emph{pseudodistribution} over solutions. We will give formal definitions in Section~\ref{subsec:toolkit}, but roughly speaking, a pseudodistribution satisfying $\calP$ is a linear functional that is indistinguishable from a distribution when evaluated on low-degree polynomials arising from the polynomials in $\calP$.

The key point then is that if one can write down a ``simple'' proof that any solution to $\calP$ has empirical mean close to $\vp$, i.e. a proof using only low-degree polynomials arising from the polynomials in $\calP$,\footnote{Practically speaking, for a proof to be ``simple'' in the above sense effectively means that the steps in the proof involve nothing more than applications of Cauchy-Schwarz and Holder's inequalities and avoid use of concentration and union bounds.} then the following learning algorithm will succeed: \begin{enumerate}
	\item[(1)] Solve an SDP to find a pseudodistribution $\psE$ satisfying $\calP$ in polynomial time.
	\item[(2)] Extract from $\psE$ an estimate for $\vp$.\footnote{We are glossing over this second step, but it turns out that a naive rounding scheme suffices (see Section~\ref{subsec:round}).}
\end{enumerate} We remark that this methodology of extracting SoS algorithms from simple proofs of identifiability has been used extensively in many recent works; we refer the reader to \cite{raghavendra2018high} for a comprehensive overview.

\subsection{Quantifying over \texorpdfstring{$\{\pm 1\}^n$}{pm1n} via Matrix SoS}
\label{subsec:quantifypm1}

We now show how to address issue \ref{issue:constraints} above. The key is to design a smaller system of polynomial constraints which imply each of the exponentially many constraints in \eqref{eq:want} under the SoS proof system, that is to say, we should be able to derive all of the constraints in \eqref{eq:want} from the constraints in the smaller system, using only ``low-degree'' steps like Cauchy-Schwarz and Holder's. We remark that although the trick we will describe for doing this has appeared previously in the literature under the name of ``matrix SoS proofs'' \cite{hopkins2018mixture}, we believe a complete but informal treatment of this technique will help the reader better appreciate the subtleties in how we extend this approach to obtain Theorem~\ref{thm:main_shape}.

To describe the trick, we first abstract out the more problem-specific details of the polynomial systems we will consider. Say we wish to encode the following exponentially large program with a smaller polynomial system.

\begin{program}{$\mathcal{Q}$}
	The variables consist of $\{Z_{\alpha,\beta}\}$ for all multisets $\alpha,\beta\subseteq[n]$ of size $t/2$, as well as some other variables $x_1,...,x_M$. The constraints include $\{p_1(x,Z)\ge 0,...,p_m(x,Z),q_1(x,Z) = 0,...,q_m(x,Z) = 0\}$ as well as the constraint \begin{equation}
		\langle Z,v^{\otimes t/2}(v^{\otimes t/2})\rangle\le 1 \ \ \ \forall v\in\{\pm 1\}^n.\label{constraint:toy}
	\end{equation}	\label{program:toy}
\end{program}

\vspace{-1.0pc}

Suppose we know that \ref{program:toy} has a satisfying assignment $(Z^*,x^*)$ to its variables--- in the systems we will actually work with, the existence of a satisfying assignment will be immediate, e.g. the set of all uncorrupted points is a satisfying assignment to the program sketched in Section~\ref{subsec:search}. 
\begin{remark}
While the meaning of $Z$ will be irrelevant to the proceeding discussion, the reader might find it helpful to think of $Z^*$, up to scaling, as the matrix $\vec{Z}[S_g]$ defined by: \begin{equation}
	\vec{Z}[S] \triangleq \frac{1}{|S|}\sum_{i\in S}\left[(X_i - \vp_i)^{\otimes t/2}\right]^{\top}\left[(X_i - \vp_i)^{\otimes t/2}\right] - \frac{1}{|S|}\sum_{i\in S}\E_{X\sim \calD_i}\left[(X - \vp_i)^{\otimes t/2}\right]^{\top}\left[(X - \vp_i)^{\otimes t/2}\right].\label{eq:Zdef}
\end{equation} The reason is that via the identity \begin{equation}
	\langle \vec{V}[S], v^{\otimes t/2}(v^{\otimes t/2})^{\top} = \frac{1}{|S|}\sum_{i\in S}\langle X_i - \vp_i,v\rangle^t - \frac{1}{|S|}\sum_{i\in S}\E_{X\sim\mathcal{D}_i}\langle X - \vp_i,v\rangle^t,
\end{equation} $\vec{Z}[S_g]$ gives a succinct way of describing the deviation of the empirical moments of the subset $S$ from the true moments.
\end{remark}
Returning to the task at hand, we would like to write down an auxiliary program $\widehat{\text{\ref{program:toy}}}$ which satisfies three criteria, namely that $\widehat{\text{\ref{program:toy}}}$ 
\begin{enumerate}

\item[(a)] has polynomially many variables and constraints

\item[(b)] implies \ref{program:toy} under the SoS proof system, and 

\item[(c)] is satisfiable.

\end{enumerate}

\noindent In this case, we would be done: we could simply solve an SDP to find a pseudodistribution $\psE$ satisfying $\widehat{\text{\ref{program:toy}}}$ and round it. Because of $(c)$ we know our SDP solver will return something, because of $(b)$ we know it will do so in polynomial time, and because of $(a)$ $\psE$ enjoys all the same properties that a pseudodistribution satisfying \ref{program:toy} would.

To see how to design such an auxiliary program $\widehat{\text{\ref{program:toy}}}$, let us suppose further that the satisfying assignment $(Z^*,x^*)$ for \ref{program:toy} satisfies the property that \eqref{constraint:toy} holds \emph{as a polynomial inequality in $v$}. Specifically, if we had formal variables $v_1,...,v_n$, suppose that one knew the existence of a proof, starting with just the polynomial equations $\{v^2_1 = 1,...,v^2_n = 1\}$ cutting out the Boolean hypercube, that the inequality $\langle Z^*,v^{\otimes t/2}(v^{\otimes t/2})^{\top}\rangle \le 1$ held, where we now view this inequality as a polynomial equation solely in the variables $v_1,...,v_n$, with coefficients specified by the fixed choice of $Z^*$.

Showing this last assumption holds in the settings we consider will be nontrivial, but assuming for now that it does, the final idea needed to write down $\hat{\text{\ref{program:toy}}}$ is the following. Instead of searching for $Z^*$ satisfying the exponentially large collection of constraints \eqref{constraint:toy}, we can search for $Z^*$ for which the abovementioned SoS proof of $\langle Z^*,v^{\otimes t/2}(v^{\otimes t/2})^{\top}\rangle \le 1$ exists. The key point is that this search problem can be encoded in a much smaller polynomial system.

In particular, as will be evident once we give formal definitions of SoS proofs, the existence of such an SoS proof is equivalent to satisfiability of some new polynomial constraints in $Z^*$ and some auxiliary variables corresponding to the steps of the SoS proof. To form $\hat{\text{\ref{program:toy}}}$, we will introduce these auxiliary variables and replace constraint \eqref{constraint:toy} with these new polynomial constraints. The reason this general approach is called ``matrix SoS'' is that these new variables will be matrix-valued, and these new constraints will be inequalities between matrix-valued polynomials. The full details of this approach are provided in Section~\ref{subsec:matrixsosproofs}.

\subsection{VC Meets Sum-of-Squares}
\label{subsec:AK}

Next, we describe the ideas that go into proving Theorem~\ref{thm:main_shape}. The first is that when $\vp$ is $(\eta,s)$-piecewise degree-$d$, to learn $\vp$ in total variation distance, it is enough to learn $\vp$ in a much weaker norm which we will denote by $\norm{\cdot}_{\calA_K}$, where $K$ is a parameter that depends on $s$ and $d$. This insight was the workhorse behind state-of-the-art density estimation algorithms for various structured univariate distribution classes \cite{diakonikolas2016learning,acharya2017sample,li2017robust}. In our setting, the main point is that if we have an estimate $\tilde{\vp}$ for $\vp$ for which $\norm{\tilde{\vp} - \vp}_{\calA_K}\le\zeta$, then by a result of \cite{acharya2017sample}, we can refine $\tilde{\vp}$ to get an estimate $\vp^*$ for which $\tvd(\vp,\vp^*) \le O(\zeta + \eta)$ efficiently. We review the details for this in a self-contained manner in Section~\ref{subsec:AK_VC}.

The algorithm of \cite{acharya2017sample} will form an important part of the boilerplate for our learning algorithm, but the key difficulty will be to actually find $\tilde{\vp}$ which is close to $\vp$ in this weaker norm. We defer definitions to Section~\ref{subsec:AK_VC}, but informally, $\norm{\vp - \tilde{\vp}}_{\calA_K}$ is small if and only if $\langle\vp - \tilde{\vp},v\rangle$ is small for all $v\in\calV^n_{K}\subset\{\pm 1\}^n$, where $\calV^n_K$ is the set of all $v\in\{\pm 1\}^n$ with at most $K$ sign changes when read as a vector from left to right (for example, $(1,1,-1,-1,1,1,1)\in\calV^7_2$).

The natural approach to do this would be to search for a $(1-\epsilon)N$-sized subset $S$ of the samples whose empirical moments satisfy \begin{equation}
	\frac{1}{|S|}\sum_{i\in S}\langle X_i - \hat{\vp}, v\rangle^t \le (8t/k)^{t/2} \ \ \ \forall v\in\calV^n_K. \label{eq:want_shape}
\end{equation} Roughly, the sample complexity savings would then come from the fact that the empirical moments will concentrate in much fewer samples because the set of directions we need to union bound over is much smaller.

Of course, if $K = O(1)$, we could afford to simply write down all $\poly(n)$ constraints in \eqref{eq:want_shape}. For typical applications of piecewise polynomial approximations though, $K$ has a logarithmic dependence on $n$, so our m ain challenge is to obtain runtimes that do not depend exponentially on $K$. In particular, just as we will use matrix SoS to succinctly encode \eqref{eq:want} for Theorem~\ref{thm:main_basic}, we will use matrix SoS to succinctly encode \eqref{eq:want_shape} for Theorem~\ref{thm:main_shape}. Next, we discuss some of the subtleties that arise in this encoding. 

\subsection{Quantifying over \texorpdfstring{$\calV^n_K$}{VnK}}
\label{subsec:quantifyVnK}

As in Section~\ref{subsec:quantifypm1}, we will abstract out the problem-specific details and focus on finding an encoding for the following program:

\begin{program}{$\mathcal{Q}'$}
	The variables consist of $\{Z_{\alpha,\beta}\}$ for all multisets $\alpha,\beta\subseteq[n]$ of size $t/2$, as well as some other variables $x_1,...,x_M$. The constraints include $\{p_1(x,Z)\ge 0,...,p_m(x,Z),q_1(x,Z) = 0,...,q_m(x,Z) = 0\}$ as well as the constraint \begin{equation}
		\langle Z,v^{\otimes t/2}(v^{\otimes t/2})^{\top}\rangle\le 1 \ \ \ \forall v\in\calV^n_K.\label{constraint:toy_shape}
	\end{equation}	\label{program:toy_shape}
\end{program}

\noindent The primary stumbling block is that, unlike the Boolean hypercube, $\calV^n_K$ is not cut out by a small number of polynomial relations. Indeed, conventional wisdom says that the sum-of-squares hierarchy is ill-suited to capturing combinatorial constraints like the ones defining $\calV^n_K$.

The first observation is that there is an alternative orthonormal basis, the \emph{Haar wavelet basis}, under which we can express any $v\in\calV^n_K$ as a vector with a small number $s = \tilde{O}(K)$ of nonzero entries. One issue with this is that $L_0$ sparsity cannot be captured by a small number of polynomial constraints, but we could try relaxing this to $L_1$ sparsity and attempt to derive an SoS proof of \eqref{eq:want_shape} out of the $L_1$ constraint.

Specifically, one could try to argue that any pseudodistribution $\psE$ over the formal variables $v_1,...,v_n,\vec{W}_1,...,\vec{W}_n$ satisfying the inequalities \begin{enumerate}
	\item[(a)] $v^2_i = 1$ for all $i\in[n]$.
	\item[(b)] $-\vec{W}_i \le (Hv)_i \le \vec{W}_i$ for all $i\in[n]$.
	\item[(c)] $\sum_i \vec{W}_i \le s$.
\end{enumerate} must satisfy \begin{equation}
	\psE\left[\left\langle Z^*,v^{\otimes t/2}(v^{\otimes t/2})^{\top}\right\rangle\right] \le 1,\label{eq:tensorwant}
\end{equation} where $Z^*$ is a constant, fixed to a satisfying assignment to $\mathcal{Q}$. Note that \eqref{eq:tensorwant} can be rewritten as \begin{equation}\left\langle Z^*,\psE\left[v^{\otimes t/2}(v^{\otimes t/2})^{\top}\right]\right\rangle \le 1,\end{equation} and one can check (Lemma~\ref{lem:contain}) that the set of all $n^{t/2}\times n^{t/2}$ matrices of the form $\psE\left[v^{\otimes t/2}(v^{\otimes t/2})^{\top}\right]$ for $\psE$ satisfying the three inequalities above is contained in the convex set $\mathcal{K}$ of all matrices whose Haar transforms are $L_{1,1}$-norm bounded\footnote{The $L_{1,1}$ norm of a matrix is defined to be the sum of the absolute values of its entries.} by $s^t$ and Frobenius norm bounded by $n^{t/2}$.

At this point it will be useful to instantiate all of this in the setting of this paper. Thinking of $Z^*$, up to scaling, as $\vec{Z}[S_g]$ as defined in \eqref{eq:Zdef}, we need to ensure that its inner product with any matrix from $\mathcal{K}$ is at most one. The matrix $\vec{Z}[S_g]$ depends on the uncorrupted samples $N$, so at this point we are merely tasked with proving some large deviation bound (where ``proof'' now is in the literal, non-SoS sense).

We expect this to hold with high probability for $N$ sublinear in $n$ because the covering number of $\mathcal{K}$ should be much smaller than that of the set of all matrices with Frobenius norm bounded by $n^t$. As covering number bounds can be quite subtle, we opt instead for a shelling argument. Specifically, we can show that any element $M$ with bounded $L_{1,1}$ and Frobenius norms can be written as a sum of $s^t$-sparse matrices whose Frobenius norms sum to at most $\norm{M}_F$ (see Lemma~\ref{lem:shelling} and its consequences in Section~\ref{encode:shape} and Appendix~\ref{app:defer}), reducing the task of building a net over $\mathcal{K}$ to building a net $\mathcal{N}$ over $s^t$-sparse matrices of Frobenius norm bounded by $n^{t/2}$.

The final and perhaps most important subtlety that arises is that as stated, this argument cannot achieve sublinear sample complexity because \emph{the inverse Haar transform of an $s^t$-sparse matrix with Frobenius norm $n^{t/2}$ may have large max-norm}, which would preclude the sorts of univariate concentration bounds one would hope to apply on each direction in $\mathcal{N}$. More concretely, the issue is that ultimately, the net $\mathcal{N}$ over $s^t$-sparse matrices of bounded Frobenius norm corresponds to a net $\mathcal{N}'$ over $\mathcal{K}$ given by the inverse Haar transform of all elements of $\mathcal{N}$. And we would need to show that for any given $M\in\mathcal{N}'$, $\langle\vec{Z}[S_g],M\rangle$ is at most one with high probability. But if we have no control over the scaling of the max-norm of these $M$'s, this is evidently impossible.

The workaround for this subtlety requires modifying the three inequalities used above, as well as the definition of $\mathcal{K}$, by incorporating properties of the Haar wavelet basis beyond just the fact that vectors from $\calV^n_K$ are sparse in this basis. Roughly speaking, the key is to exploit the inherent multi-scale nature of the Haar wavelet basis.

This is best understood with an example. Instead of matrices, we will work with vectors (the reader can think of this as the ``$t = 1$'' case). In the following example, we will first try to convey 1) that there exist sparse vectors with $L_2$ norm $\sqrt{n}$ but whose inverse Haar transforms are as large as $\sqrt{n/2}$ in $L_{\infty}$ norm. To reiterate, this is an issue because any $w\in\R^n$ which is a Haar transform of some vector $v\in\{\pm 1\}^n$ with few sign changes is sparse and has $L_2$ norm $\sqrt{n}$, yet the inverse Haar transform of $w$, i.e. $v$ itself, has $L_{\infty}$ norm 1. In other words, simply relaxing the set of $v\in\{\pm 1\}^n$ to the set of all vectors whose Haar transforms are sparse introduces problematic new vectors with substantially different properties than the vectors $v$. We will then 2) give a flavor of how we circumvent this crucial subtlety.

\begin{example}
	Let $n = 2^m$. The Haar wavelet basis for $\R^n$ contains the vector \begin{equation}\psi_{\ell}\triangleq \left(\frac{1}{\sqrt{2}},-\frac{1}{\sqrt{2}},0,0,...,0\right).\end{equation} Say this is the $\ell$-th vector in the basis. Then the vector $w$ which has $\ell$-th entry equal to $\sqrt{n}$ and all other entries 0 is clearly sparse and has $L_2$ norm $\sqrt{n}$. But its inverse Haar transform is \begin{equation}\left(\sqrt{n/2},-\sqrt{n/2},0,0,...,0\right),\end{equation} which has largest entry $\sqrt{n/2}$, whereas obviously any $v\in\{\pm 1\}^n$ has largest entry $1$.

	One reason this example is not so bad is that if we express any $v\in\{\pm 1\}^n$ as a linear combination of Haar wavelets, the coefficient for the $\ell$-th Haar wavelet, by orthonormality of the Haar wavelet basis, is $\langle v,\psi_{\ell}\rangle \le \sqrt{2}$. That is, the Haar transform of any such $v$ has $\ell$-th entry at most $\sqrt{2}$. So if we added to the collection of constraints defining $\mathcal{K}$ this additional constraint, we would already get rid of some problematic vectors like $w$.\label{example:problematic}
\end{example}

More generally, problematic vectors like $w$ in Example~\ref{example:problematic} exist at every ``level'' of the Haar wavelet basis, and it will be necessary to handle each of these levels appropriately. We defer the details to Lemma~\ref{lem:haar} and its consequences in Sections~\ref{encode:shape} and Appendix~\ref{app:defer}.

\section{Technical Preliminaries}
\label{sec:technical_prelims}

\subsection{Notation}

\begin{itemize}[leftmargin=*]
	\item Let $\Delta^n\subset\R^n$ be the simplex of nonnegative vectors whose coordinates sum to 1. Any $p\in\Delta^n$ naturally corresponds to a probability distribution over $[n]$.


	\item Given $\vp\in\Delta^n$, let $\mul_k(\vp)$ denote the distribution over $\Delta^n$ given by sampling a frequency vector from the multinomial distribution arising from $k$ draws from the distribution over $[n]$ specified by $\vp$, and dividing by $k$.

	For example, when $n = 2$ and $\vp = (p,1-p)$, $\mul_k(\vp)$ is the distribution given by sampling from $\Bin(k,p)$ and dividing by $k$.

	\item Given polynomials $p,q_1,...,q_m$ in formal variables $x_1,...,x_n$, we say that $p$ is in the ideal generated by $q_1,...,q_m$ at degree $d$ if there exist polynomials $\{s_i\}_{i\in[m]}$ for which $q(x) = \sum^m_{i=1} s_i(x) q_i(x)$ where each $s_i(x)q_i(x)$ is of degree at most $d$.

	\item Recall the definition of the flattened tensor from \eqref{eq:Zdef}. For any $S\subseteq[N]$, \begin{equation}
	\vec{Z}[S] \triangleq \frac{1}{|S|}\sum_{i\in S}\left[(X_i - \vp_i)^{\otimes t/2}\right]^{\top}\left[(X_i - \vp_i)^{\otimes t/2}\right] - \frac{1}{|S|}\sum_{i\in S}\E_{X\sim \calD_i}\left[(X - \vp_i)^{\otimes t/2}\right]^{\top}\left[(X - \vp_i)^{\otimes t/2}\right].\label{eq:Zdefagain}
	\end{equation}

	\item Given matrix $\vec{M}$, denote by $\norm{\vec{M}}_{1,1}$ the sum of the absolute values of its entries.
\end{itemize}

\subsection{The Generative Model}

Throughout the rest of the paper, let $\epsilon, \delta > 0$, $n,k,N\in\N$, and let $\vp\in\Delta^n$ be some probability distribution over $[n]$.

\begin{defn}
	We say $Z^1,...,Z^N$ is an \emph{$\epsilon$-corrupted set of $N$ $\delta$-diverse batches of size $k$ from $\vp$} if they are generated via the following process:

	\begin{itemize}
	 	\item For every $i\in[(1-\epsilon)N]$, $\tilde{Z}^i = (\tilde{Z}^i_1,...,\tilde{Z}^i_k)$ is a set of $k$ iid draws from $\vp_i$, where $\vp_i\in\Delta^n$ is some probability distribution over $[n]$ for which $\tvd(\vp,\vp_i) \le \delta$.
	 	\item A computationally unbounded adversary inspects $\tilde{Z}^1,...,\tilde{Z}^{(1-\epsilon)N}$ and adds $\epsilon N$ arbitrarily chosen tuples $\tilde{Z}^{(1-\epsilon)N+1},...,\tilde{Z}^N\in[n]^k$, and returns the entire collection of tuples in any arbitrary order as $Z^1,...,Z^N$.
 	\end{itemize}
\end{defn}

\subsection{Sum-of-Squares Toolkit}
\label{subsec:toolkit}

Let $x_1,...,x_n$ be formal variables, and let Program $\calP$ be a set of polynomial equations and inequalities $\{p_1(x)\ge 0,...,p_m(x)\ge 0,q_1(x) = 0,...,q_m(x) = 0\}$.

We say that the inequality $p(x)\ge 0$ has a degree-$d$ SoS proof using $\calP$ if there exists a polynomial $q(x)$ in the ideal generated by $q_1(x),...,q_m(x)$ at degree $d$, together with sum-of-squares polynomials $\{r_S(x)\}_{S\subseteq[m]}$ (where the index $S$ is a multiset), such that \begin{equation}
	p(x) = q(x) + \sum_{S\subseteq[m]}r_S(x)\cdot\prod_{i\in S}p_i(x),
\end{equation} and such that $\deg(r_S(x)\cdot\prod_{i\in S}p_i(x))\le d$ for each multiset $S\subseteq[m]$. We denote this by the notation \begin{equation}
	\calP \sos{d} p(x)\ge 0
\end{equation} When $\calP = \{1\}$, we will denote this by $\sos{d} p(x)\ge 0$.

A fact we will use throughout without comment is that SoS proofs compose well:

\begin{fact}
	If $\calP\sos{d} p(x)\ge 0$ and $\calB\sos{d'} q(x)\ge 0$, then $\calP\cup\calB\sos{\max(d,d')} p(x) + q(x)\ge 0$ and $\calP\cup\calB\sos{dd'} p(x)q(x)\ge 0$.
\end{fact}

It is useful to consider the objects dual to SoS proofs, namely pseudodistributions. A degree-$d$ pseudodistribution is a linear functional $\psE: \R[x]_{\le d}\to \R$ satisfying the following properties: \begin{enumerate}
	\item Normalization: $\psE[1] = 1$
	\item Positivity: $\psE[p(x)^2]$ for every $p$ of degree at most $d/2$.
\end{enumerate} We will use the terms ``pseudodistribution'' and ``pseudoexpectation'' interchangeably.

A degree-$d$ pseudodistribution $\psE$ \emph{satisfies} Program $\calP = \{p_1(x)\ge 0,...,p_m(x)\ge 0,q_1(x) = 0,...,q_m(x) = 0\}$ if for every multiset $S\subseteq[m]$ and sum-of-squares polynomial $r(x)$ for which $\deg(r(x)\cdot\prod_{i\in S}p_i(x))\le d$, we have $\psE[r(x)\cdot\prod_{i\in S}p_i(x)] \ge 0$, and for every $q(x)$ in the ideal generated by $q_1,...,q_m$ at degree $d$, we have $\psE[q(x)] = 0$.

The following fundamental fact is a consequence of SDP duality:

\begin{fact}
	If $\calP\sos{d} p(x)\ge 0$ and $\psE$ is a degree-$d$ pseudodistribution satisfying $\calP$, then $\psE$ satisfies $\calP\cup \{p\ge 0\}$.
\end{fact}

We collect some basic inequalities that are captured by the SoS proof system, the proofs of which can be found, e.g., in Appendix A of \cite{hopkins2018mixture} and \cite{ma2016polynomial}.

\begin{fact}[SoS Cauchy-Schwarz]
	Let $x_1,...,x_n,y_1,...,y_n$ be formal variables. Then \begin{equation}
		\sos{4} \left(\sum^n_{i=1}x_iy_i\right)^2 \le \left(\sum^n_{i=1}x^2_i\right)\cdot\left(\sum^n_{i=1}y^2_i\right)
	\end{equation}
\end{fact}

\begin{fact}[SoS Holder's]
	Let $w_1,...,w_n,x_1,...,x_n$ be formal variables. Then for any $t\in\N$ a power of 2, we have \begin{equation}
		\{w^2_i = w_i \ \forall i\in[n]\} \sos{O(t)} \left(\sum^n_{i=1}w_ix_i\right)^t \le \left(\sum^n_{i=1}w_i\right)^{t-1}\cdot\sum^n_{i=1}x^t_i
	\end{equation} and \begin{equation}
		\{w^2_i = w_i \ \forall i\in[n]\} \sos{O(t)} \left(\sum^n_{i=1}w_ix_i\right)^t \le \left(\sum^n_{i=1}w_i\right)^{t-1}\cdot\sum^n_{i=1}w_i x^q_i.
	\end{equation}
\end{fact}

We will also use the following consequence of scalar Holder's inequality.

\begin{fact}
	Let $\ell(x)$ be a linear form in the formal variables $x_1...,x_n$. Then if $\psE$ is a degree-$t$ pseudodistribution, then \begin{equation}
		\psE[\ell(x)]^t \le \psE[\ell(x)^t].
	\end{equation}\label{fact:scalarholders}
\end{fact}

\begin{proof}
	Because $\psE$ is a degree-$t$ pseudodistribution, there exists a pseudo-density $H(\cdot)$ such that $\psE[p(x)] = \sum_x H(x)\cdot p(x)$ for any degree-$t$ polynomial $p$. So by scalar Holder's inequality we get that \begin{equation}
		\psE[\ell(x)]^t = \left(\sum_x H(x)\ell(x)\right)^t \le \left(\sum_x H(x)\right)^{t-1} \cdot \left(\sum_x H(x)\ell(x)^t\right) = \psE[1]^{t-1}\cdot\psE[\ell(x)^t] = \psE[\ell(x)^t]
	\end{equation} as claimed.
\end{proof}

The following elementary inequality will also be useful.

\begin{fact}
	$\{x^2 = 1\} \sos{2} -1\le x\le 1$.\label{fact:1bounded}
\end{fact}

\begin{proof}
	Noting that \begin{equation}1 - x = \frac{1}{2}(1 - x^2 + (x-1)^2) \ \ \ \text{and} \ \ \ 1 + x = \frac{1}{2}(1 - x^2 + (x+1)^2),\label{eq:trivial}\end{equation} the claim follows.
\end{proof}

For a thorough treatment on the SoS proof system, we refer the reader to \cite{o2013approximability,barak2014sum}.

\subsection{Certifiably Bounded Distributions}

Recall from Section~\ref{subsec:quantifypm1} that a prerequisite for the ``matrix SoS'' approach to work is that the exponentially large program from Section~\ref{subsec:search} must have a satisfying assignment for which there exists an SoS proof of the requisite empirical moment bounds \eqref{eq:want} using the axioms $\{v^2_i = 1 \ \forall i\in[n]\}$. A necessary condition for this to hold is for there to be an SoS proof from these axioms that the true moments of $\vp$ itself satisfy these same bounds. Again, we emphasize that these bounds should be regarded as polynomial inequalities solely in the variables $v_1,...,v_n$.

Here we formalize what we mean by the existence of such a proof.


\begin{defn}
	A distribution $\calD$ over $\R^d$ with mean $\mu$ is $(t,\infty)$-explicitly bounded with variance proxy $\sigma$ if for every even $2\le s\le t$: \begin{equation}
		\{v^2_i = 1 \ \forall i\in[n]\} \sos{s} \E_{Y\sim\calD}[\langle Y - \mu,v\rangle^s] \le (\sigma s)^{s/2} \label{eq:explicitlybounded}
	\end{equation}
\end{defn}

We remark that while a consequence of Theorem~\ref{thm:latala}, due to \cite{latala1997estimation}, is that the moments of any multinomial distribution satisfy these bounds, the proof in that work uses exponentials and is thus not an SoS proof without additional modifications to the argument. Here we give an SoS proof, at the cost of less desirable constants than those of \cite{latala1997estimation}. To our knowledge, this SoS proof is new.

\begin{lem}
	Let $\calD = \mul_k(\vp)$ for any $\vp\in\Delta^n$. Then $\calD$ is $(k,\infty)$-explicitly bounded with variance proxy $8/k$.\label{lem:basicbound}
\end{lem}

\begin{proof}
	It is enough to show \eqref{eq:explicitlybounded} for $v$ for which $\norm{v}_{\infty}=1$. By definition $\mu = \E_{Y\sim\calD}[Y]$, so we may symmetrize as follows: \begin{align}
		\sos{s} \E_{Y\sim\cal D}[\langle Y - \mu,v\rangle^s] &= \E_{Y\sim\calD}[\langle Y - \E_{Y'\sim\calD}[Y'], v\rangle^s] \\
		&\le \E_{Y,Y'\sim\calD}[\langle Y - Y', v\rangle^s],
	\end{align} where the inequality follows from SoS Cauchy-Schwarz. But note that the random variable $\langle Y,v\rangle$ is the average of $k$ independent copies of the random variable which takes on value $v_i$ with probability $p_i$ for every $i\in[n]$. So define $Z$ to be the symmetric random variable which takes on value $(v_i - v_{i'})$ with probability $p_ip'_i$ for every $(i,i')\in[n]\times[n]$. Then for $Z_1,...,Z_k$ independent copies of $Z$, \begin{equation}\langle Y - Y',v\rangle \eqd \frac{1}{k}\sum Z_i\end{equation}

	We conclude that for any $1\le s\le k$, \begin{align}
		\sos{s} \E_{Y\sim\calD}[\langle Y - \mu,v\rangle^s] &\le \frac{1}{k^s}\E[(Z_1 + \cdots + Z_k)^s] \\
		&= \frac{1}{k^s}\sum_{\beta: |\beta| = s}\binom{s}{\beta_1,...,\beta_k}\E[Z_{\beta}] \label{eq:betas}\\
		&= \frac{1}{k^s}\sum_{\substack{\beta: |\beta| = s \\ \beta_i \ \text{even} \ \forall 1\le i\le k}}\binom{s}{\beta_1,...,\beta_k} \E[Z_{\beta}] \label{eq:odddie}\\
		&\le \frac{1}{k^s}(2sk)^{s/2}\cdot\max_{\beta}\E[Z_{\beta}] \label{eq:ballsbins}\\
		&\le (2s/k)^{s/2}\cdot\max_{\beta}\prod^k_{i=1}\E[Z^{\beta_i}_i] \label{eq:indep}\\
		&\le (8s/k)^{s/2} \label{eq:finalmomentbound},
	\end{align} where the sum in \eqref{eq:betas} ranges over all monomials $\beta$ of total degree $s$, that is, all tuples $\beta\in[s]^k$ for which $\sum^k_{i=1} \beta_i = s$. Equation \eqref{eq:odddie} follows from the fact that $\E[Z_{\beta}] = \prod^k_{i=1}\E[Z^{\beta_i}_i]$ by independence, and $\E[Z^d_i] = 0$ for any odd $d$ because $Z$ is symmetric. For equation \eqref{eq:ballsbins}, note that by balls-and-bins, there are $\binom{s/2 + k - 1}{s/2}\le \left(\frac{3ek}{s}\right)^{s/2}$ choices of $\beta$, and $\binom{s}{\beta_1,...,\beta_k}\le s!\le s^{s+1/2}e^{-s + 1}$, and we may crudely bound the product of these quantities as \begin{equation}
		(3ek/s)^{s/2} \cdot s^{s+1/2}e^{-s + 1} \le (2sk)^{s/2}.
	\end{equation} Equation \eqref{eq:indep} follows by independence, and for \eqref{eq:finalmomentbound} we need that for every even $2\le d\le s$, there is a degree-$s$ SoS proof that $\E[Z^d]\le 2^d$. But by Fact~\ref{fact:1bounded}, $\{v^2_i = 1 \ \forall i\in[n]\} \sos{2} -2 \le v_i - v_{i'}\le 2$, from which there is a degree-$d$ proof that $(v_i - v_{i'})^d \le 2^d$. So \begin{equation}
		\{v^2_i = 1 \ \forall i\in[n]\} \sos{d} \E[Z^d] = \sum_{i,i'}p_ip_{i'}(v_i - v_{i'})^d \le 2^d\sum_{i,i'}p_ip_{i'} = 2^d
	\end{equation} as claimed.
\end{proof}


\section{Efficiently Learning from Untrusted Batches}
\label{sec:batches}

In this section we prove our result on the general problem of learning from untrusted batches.

\begin{thm}
	Let $t\ge 4$ be any integer. There is an algorithm that draws an $\epsilon$-corrupted set of $N$ $\delta$-diverse batches of size $k$ from $\vp$ for $N\ge \delta^{-2}\epsilon^{-2} n^{O(t)}\cdot k^t/t^{t-1}$, runs in time $\delta^{-2t}\epsilon^{-2t} n^{O(t^2)}\cdot k^{t^2}/t^{t(t-1)}$, and with probability $1 - 1/\poly(n)$ outputs a distribution $\hat{\vp}$ for which $\tvd(\vp,\hat{\vp}) \le O(\delta + \epsilon^{1-1/t}\sqrt{t/k})$.
	\label{thm:main_basic}
\end{thm}

We will describe our polynomial system and algorithm, list deterministic conditions under which our algorithm will succeed, give an SoS proof of identifiability, and conclude the proof of Theorem~\ref{thm:main_basic} by analyzing the rounding step of our algorithm. We will defer technical details for how to encode some of the constraints of our polynomial system to Section~\ref{sec:matrixsos}.

\subsection{An SoS Relaxation}

Let $t$ be a power of two, to be chosen later. For $\vp\in\Delta^n$, let $\calD = \mul_k(\vp)$. Let $Y_1,...,Y_N\in\Delta^n$ be the set of iid samples from $\calD_1,...,\calD_N$ respectively, where for each $i\in[N]$ we have $\calD_i = \mul_k(\vp_i)$ for some $\vp_i\in\Delta^n$ satisfying $\tvd(\vp_i,\vp)\le\delta$. Let $\{X_i\}_{i\in[N]}\in\Delta^n$ be those samples after an $\epsilon$-fraction have been corrupted. Let $S_g\subset[N]$ (resp. $S_b\subset[N]$) denote the subset of uncorrupted (resp. corrupted) points.

\begin{program}{$\calP$}
The variables are $\{w_i\}_{i\in[N]}$, $\{\hat{\vp}_i\}_{i\in[N]}$, and $\hat{\vp}$, and the constraints are \begin{enumerate}
	\item $w^2_i = w_i$ for all $i\in[N]$.
	\item $\sum w_i = (1 - \epsilon)N$.
	\item For every $v\in\{\pm 1\}^n$ and every $i\in[N]$, $\langle\hat{\vp}_i - \hat{\vp},v\rangle \le 5\delta$.\label{constraint:l1}
	\item $\sum_{i\in[N]} w_i X_i = \hat{\vp}\sum_{i\in[N]}w_i$.
	\item For every $v\in\{\pm 1\}^n$ \begin{equation}
		\sum_{i\in[N]}w_i\langle X_i - \hat{\vp}_i,v\rangle^t \le (8t/k)^{t/2}\cdot \sum_{i\in[N]}w_i\label{eq:mainmomentbound}
	\end{equation}\label{constraint:momentbound}
	\item $\hat{\vp}_i \ge 0$ for all $i\in[n]$ and $\sum_i \hat{\vp}_i = 1$. \label{constraint:simplex}
\end{enumerate}\label{program:basic}
\end{program}

Note that constraints \eqref{constraint:l1} and \eqref{constraint:momentbound} are quantified over all $v\in\{\pm 1\}^n$, so as stated, Program~\ref{program:basic} is a system of exponentially many polynomial constraints. In Section~\ref{sec:matrixsos}, we will explain how to encode these constraints as a small system of polynomial constraints. For now, we state the following without proof.

\begin{lem}
There is a system \ref{program:matrixsos_basic} of degree-$O(t)$ polynomial equations and inequalities in the variables $\{w_i\}$, $\{\hat{\vp}_i\}$, $\hat{\vp}$, and $n^{O(t)}$ other variables, whose coefficients depend on $\epsilon, t, X_1,...,X_n$ such that

\begin{enumerate}
	\item (Satisfiability) With probability at least $1- 1/\poly(n)$, \ref{program:matrixsos_basic} has a solution in which $\hat{\vp} = \vp$ and for each $i\in[N]$, $\hat{\vp}_i = \vp_i$ and $w_i$ is the indicator for whether $X_i$ is an uncorrupted point.
	\item (Encodes Moment Bounds) $\text{\ref{program:matrixsos_basic}} \sos{O(t)} \text{\ref{program:basic}}$.
	\item (Solvability) If \ref{program:matrixsos_basic} is satisfied, then for every integer $C > 0$, there is an $n^{O(Ct)}$-time algorithm which outputs a degree-$Ct$ pseudodistribution which satisfies \ref{program:matrixsos_basic} up to additive error $2^{-n}$.
\end{enumerate}
\label{lem:satisfied}
\end{lem}

This suggests the following algorithm for learning from untrusted batches: use semidefinite programming to efficiently obtain a pseudodistribution over solutions to Program~\ref{program:matrixsos_basic}, and round this pseudodistribution to an estimate for $\vp$ by computing the pseudoexpectation of the $\hat{\vp}$ variable. A formal specification of this algorithm, which we call \textsc{LearnFromUntrusted}, is given in Algorithm~\ref{alg:basic} below.

\begin{center}
	\myalg{alg:basic}{LearnFromUntrusted}{
		\textbf{Input}: Corruption parameter $\epsilon$, diversity parameter $\delta$, support size $n$, batch size $k$, samples $\{X_i\}_{i\in[N]}$, degree $t$ \\
		\textbf{Output}: Estimate $\hat{\vp}$
		\begin{enumerate}
			\item Run SDP solver to find a pseudodistribution $\psE$ of degree $O(t)$ satisfying the constraints of Program~\ref{program:matrixsos_basic}.
			\item Return $\psE[\hat{\vp}]$.
		\end{enumerate}
	}
\end{center}

\begin{remark}
Here we clarify some points regarding numerical accuracy of \textsc{LearnFromUntrusted} and the other algorithms presented in this work. Formally, the pseudodistribution computed by \textsc{LearnFromUntrusted} satisfies the constraints of Program~\ref{program:matrixsos_basic} to precision $2^{-n}$ in the sense that for any sum-of-squares $q$ and constraint polynomials $f_1,...,f_{\ell}\in\text{\ref{program:matrixsos_basic}}$ for which $\deg(q\cdot \prod_{i\in[\ell]}f_i)\le O(t)$, we have that $\psE\left[q\cdot \prod_{i\in[\ell]}f_i\right] \ge -2^{-n}\norm{q}_2$, where $\norm{q}_2$ denotes the $L_2$ norm of the vector of coefficients of $q$. On the other hand, in our analysis, we show that \ref{program:matrixsos_basic}$\sos{O(t)}$\ref{program:basic} and then argue using the constraints of \ref{program:basic} instead. But because the coefficients in the SoS proof that \ref{program:matrixsos_basic}$\sos{O(t)}$\ref{program:basic} are polynomially bounded, the pseudodistribution computed by \textsc{LearnFromUntrusted} also satisfies the constraints of Program~\ref{program:basic} to precision $2^{-\Omega(n)}$, which will be sufficient for the simple rounding we analyze in Section~\ref{subsec:round}.
\end{remark}

\subsection{Deterministic Conditions}

We will condition on the following deterministic conditions holding simultaneously: \begin{enumerate}[label=(\Roman*)]
	\item \label{enum:satisfied} The ``Satisfiability'' condition of Lemma~\ref{lem:satisfied} holds.
	\item \label{enum:meanconc} The mean of the uncorrupted points concentrates: \begin{equation}
		\norm{\frac{1}{N}\sum_{i\in S_g}(X_i - \vp_i)}_{1}\le O(\delta + \epsilon^{1-1/t}\sqrt{t/k}) \label{eq:uncorruptconc}
	\end{equation}
	\item \label{enum:tensorconc} The empirical $t$-th moments concentrate: \begin{equation}
		\{v^2_i = 1 \ \forall i\in[n]\} \sos{t} \frac{1}{N}\sum_{i\in[N]}\langle Y_i - \vp_i, v\rangle^t - \frac{1}{N}\sum_{i\in[N]}\E_{Y_i\sim\calD_i}\langle Y_i - \vp_i, v\rangle^t \le (8t/k)^{t/2}
	\end{equation}
\end{enumerate}

\begin{lem}
	Conditions \ref{enum:satisfied}, \ref{enum:meanconc}, \ref{enum:tensorconc} hold simultaneously with probability $1 - 1/\poly(n)$.\label{lem:conc}
\end{lem}

We first need the following elementary concentration inequalities.

\begin{fact}
	If $Y_1,...,Y_N$ are drawn from $\mul_k(\vp_1),...,\mul_k(\vp_N)$ respectively, then \begin{equation}
		\Pr\left[\norm{\frac{1}{N}\sum_{i\in[N]}Y_i - \frac{1}{N}\sum_{i\in[N]}\vp_i}_1 > \epsilon\right] \le n\cdot e^{-2\epsilon^2N/n^2}
	\end{equation} \label{fact:multconc}
\end{fact}

\begin{proof}
	Note that for each $i\in[N]$, $j\in[n]$, $(Y_i)_j$ is distributed as $\Ber((\vp_i)_j)$. So by Hoeffding's inequality, \begin{equation}\Pr\left[\left|\frac{1}{N}\sum_{i\in[N]}(Y_i)_j - \frac{1}{N}\sum_{i\in[N]}(\vp_i)_j\right| \ge \eta \right] \le e^{-2N\eta^2}.\end{equation} The claim follows by taking $\eta = \epsilon/n$ and union bounding over $j$.
\end{proof}

\begin{fact}
Let $Y_1,...,Y_N$ be independent samples from $\calD_1,...,\calD_N$. For every $i\in[N]$, define $Z_i = Y_i - \vp_i$. If $N\ge \Omega(t\cdot(k/8t)^t\cdot n^{2t}\log^2(n))$, then with probability $1 - 1/\poly(n)$ the following holds: for every multi-index $\theta\in[t]^{n}$ for which $\sum\theta_i = t$ we have that \begin{equation}
	\left|\frac{1}{N}\sum_{i\in[N]}Z_i^{\theta} - \frac{1}{N}\sum_{i\in[N]}\E_{Z\sim\calD_i - \vp_i}[Z^{\theta}]\right| \le n^{-t}\cdot (8t/k)^{t/2}.
\end{equation}\label{fact:monombound}
\end{fact}

\begin{proof}
	Note that because $Y_i,\vp_i\in[0,1]$, the random variables $Z_i^{\theta}$ only take values within $[-1,1]$. By Hoeffding's inequality, \begin{equation}
		\Pr\left[\left|\frac{1}{N}\sum_{i\in[N]}Z_i^{\theta} - \frac{1}{N}\sum_{i\in[N]}\E_{Z\sim\calD_i - \vp_i}[Z^{\theta}]\right| \ge \eta\right] \le 2e^{-N\eta^2/2},\label{eq:hoeffding}
	\end{equation} so the lemma follows by by taking $\eta = n^{-t}\cdot (8t/k)^{t/2}$ and union-bounding over all $n^t$ choices of $\theta$.
\end{proof}

\begin{proof}[Proof of Lemma~\ref{lem:conc}]
	(I) holds with probability at least $1 - 1/\poly(n)$ according to Lemma~\ref{lem:satisfied}.

	Because $\{Y_i\}_{i\in S_g}$ are independent draws from $\{\vp_i\}_{i\in S_g}$, (II) holds with probability at least $1 - 1/\poly(n)$ provided $N\ge \Omega((k/t)n^2\log^2 n \cdot\delta^{-2}\epsilon^{-2})$, according to Fact~\ref{fact:multconc}.

	Finally, we verify (III) holds with high probability. For every $i\in[N]$ define $Z_i = Y_i - \vp_i$. The inequality we would like to exhibit an SoS proof for is equivalent to the inequality \begin{equation}
		\left|\sum_{\theta,\theta': |\theta| = |\theta'| = t/2}v_{\theta}v_{\theta'}\left(\frac{1}{N}\sum_{i\in[N]}Z_i^{\theta}Z_i^{\theta'} - \frac{1}{N}\sum_{i\in[N]}\E_{Z\sim\calD_i - \vp_i}\left[Z^{\theta}Z^{\theta'}\right]\right)\right| \le (8t/k)^{t/2},\label{eq:iiiequiv}
	\end{equation} where $v_{\theta}\triangleq \prod_{i\in\theta}v_i$. Note that \begin{equation}
		\{v^2_i = 1 \ \forall i\in[n]\} \sos{t} -1 \le v_{\theta}v_{\theta'} \le 1,
	\end{equation} If the outcome of Fact~\ref{fact:monombound} holds for all $n^t$ monomials of the form $\theta\cup\theta'$, then there is a degree-$t$ proof, using the axioms $\{v^2_i = 1\ \forall i\in[n]\}\sos{t}$, that \eqref{eq:iiiequiv} holds. We conclude that (III) holds with probability $1 - 1/\poly(n)$.

	By a union bound over all events upon which we conditioned, we conclude that (I), (II), (III) are simultaneously satisfied with probability $1 - 1/\poly(n)$.
\end{proof}

\subsection{Identifiability}
\label{subsec:basic_id}

The key step is to give an SoS proof of identifiability. In other words, we must demonstrate in the SoS proof system that the constraints of Program~\ref{program:basic} imply that $\hat{\vp}$ is sufficiently close to $\vp$. The main claim in this section is the following.

\begin{lem}
	Suppose Conditions \ref{enum:satisfied}-\ref{enum:tensorconc} hold. Then for any $v\in\{\pm 1\}^n$, we have that \begin{equation}
		\text{\ref{program:basic}} \sos{O(t)} \langle\hat{\vp} - \vp, v\rangle^{t} \le O\left(\delta^t + \epsilon^{t-1}(t/k)^{t/2}\right).\label{eq:main_id}
	\end{equation}\label{lem:main}
\end{lem}

First note that for any $i\in[N]$, \begin{align}
	\sum_{i\in[N]} w_i\langle\hat{\vp} - \vp, v\rangle &= \sum_{i\in[N]} w_i\langle \hat{\vp} - \vp, v\rangle \nonumber\\
	&= \sum_{i\in [N]}w_i\langle \hat{\vp} - \vp_i,v\rangle + \sum_{i\in[N]}w_i\langle \vp_i - \vp,v\rangle \nonumber\\
	&\le \sum_{i\in [N]}w_i\langle \hat{\vp}- \vp_i,v\rangle + 2N\delta,\label{eq:split}
\end{align} where the inequality follows from the assumption that $\tvd(\vp,\vp_i)\le\delta$. We bound the former term in \eqref{eq:split}: \begin{align}
	\sum_{i\in[N]}w_i\langle \hat{\vp} - \vp_i,v\rangle &= \sum_{i\in[N]}w_i\langle X_i - \vp_i,v\rangle \nonumber\\
	&= \sum_{i\in S_g}\langle X_i - \vp_i, v\rangle + \sum_{i\in S_g}(w_i - 1)\langle X_i - \vp_i, v\rangle + \sum_{i\in S_b}w_i \langle X_i - \vp_i,v\rangle \nonumber\\
	&= \sum_{i\in S_g}\langle X_i - \vp_i, v\rangle + \sum_{i\in S_g}(w_i - 1)\langle X_i - \vp_i, v\rangle + \nonumber\\ 
	& \ \ \ \ \ \  \sum_{i\in S_b}w_i\langle X_i - \hat{\vp}_i, v\rangle + \sum_{i\in S_b}w_i\langle \hat{\vp} - \vp_i, v\rangle + \sum_{i\in S_b}w_i\langle\hat{\vp}_i - \hat{\vp},v\rangle, \nonumber\\
	&\le \sum_{i\in S_g}\langle X_i - \vp_i, v\rangle + \sum_{i\in S_g}(w_i - 1)\langle X_i - \vp_i, v\rangle + \nonumber\\ 
	& \ \ \ \ \ \  \sum_{i\in S_b}w_i\langle X_i - \hat{\vp}_i, v\rangle + \sum_{i\in S_b}w_i\langle \hat{\vp} - \vp_i, v\rangle + 5N\epsilon\delta\label{eq:bigsplit}
\end{align} where the inequality follows from Constraint~\ref{constraint:l1} of $\text{\ref{program:basic}}$. This rearranges to \begin{equation}\sum_{i\in S_g}w_i\langle \hat{\vp} - \vp_i, v\rangle\le 5N\epsilon\delta + \sum_{i\in S_g}\langle X_i - \vp_i, v\rangle + \sum_{i\in S_g}(w_i - 1)\langle X_i - \vp_i, v\rangle + \sum_{i\in S_b}w_i\langle X_i - \hat{\vp}_i, v\rangle.\end{equation} Taking the $t$-th power of both sides of \eqref{eq:split} and invoking \eqref{eq:bigsplit} and the inequality $\sos{t} (a + b + c + d + e)^t\le \exp(t)(a^t + b^t + c^t + d^t + e^t)$, we conclude that \begin{dmath}
	\text{\ref{program:basic}} \sos{t} \left(\sum_{i\in S_g}w_i\right)^t\langle \hat{\vp} - \vp, v\rangle^t \le \exp(t)\left[(N(2+5\epsilon)\delta)^t+ \underbrace{\left(\sum_{i\in S_g}\langle X_i - \vp_i, v\rangle\right)^t}_{\text{Lemma~\ref{lem:A}}} + \underbrace{\left(\sum_{i\in S_g}(w_i - 1)\langle X_i - \vp_i, v\rangle\right)^t}_{\text{Lemma~\ref{lem:B}}} + \underbrace{\left(\sum_{i\in S_b}w_i \langle X_i - \hat{\vp}_i, v\rangle\right)^t}_{\text{Lemma~\ref{lem:C}}}\right],\label{eq:lemmasplit}
\end{dmath} which we bound using Lemmas~\ref{lem:A}, \ref{lem:B}, and \ref{lem:C} below. Intuitively, the term for Lemma~\ref{lem:A} corresponds to sampling error from uncorrupted samples from $\calD$, the term for Lemma~\ref{lem:B} corresponds to the possible failure of the subset selected by $w_i$ to capture some small fraction of the uncorrupted samples, and the term for Lemma~\ref{lem:C} corresponds to the error contributed by the adversarially chosen vectors.

\begin{lem}
	Suppose Conditions \ref{enum:satisfied}-\ref{enum:tensorconc} hold. Then for any $v\in\{\pm 1\}^n$, we have that \begin{equation}
		\text{\ref{program:basic}} \sos{O(t)}\left(\sum_{i\in S_g}\langle X_i - \vp_i, v\rangle\right)^t \le O(N)^t\cdot(\delta^t + \epsilon^{t-1}\cdot (t/k)^{t/2}). \label{eq:A}
	\end{equation}
	\label{lem:A}
\end{lem}

\begin{proof}
	By SoS Holder's, we have that \begin{align}
		\text{\ref{program:basic}} \sos{O(t)}\left(\sum_{i\in S_g}\langle X_i - \vp_i, v\rangle\right)^t &= \left\langle \sum_{i\in S_g}(X_i - \vp_i),v\right\rangle^t \\
		&\le \norm{\sum_{i\in S_g}(X_i - \vp_i)}^t_1 \\
		&\le \left(N\cdot O(\delta + \epsilon^{1-1/t}\cdot \sqrt{t/k})\right)^t \\
		&\le O(N)^t\cdot(\delta^t + \epsilon^{t-1}\cdot (t/k)^{t/2})
	\end{align} as claimed, where the penultimate step follows by (II) and the last step follows by (scalar) Holder's.
\end{proof}

For Lemma~\ref{lem:B}, we will use the following helper lemma.

\begin{lem}
	Suppose Condition~\ref{enum:tensorconc} holds. Then for any $v\in\{\pm 1\}^n$, we have that \begin{equation}
		\text{\ref{program:basic}} \sos{O(t)}\sum_{i\in[N]}\langle Y_i - \vp_i, v\rangle^t \le 2N(8t/k)^{t/2}.\label{eq:nesterov}
	\end{equation}
	\label{lem:nesterov}
\end{lem}

\begin{proof}
	By Lemma~\ref{lem:conc} and Lemma~\ref{lem:basicbound}, we have that \begin{equation}
		\sum_{i\in[N]}\langle Y_i - \vp_i, v\rangle^t \le N\cdot (8t/k)^{t/2} + \sum_{i\in[N]}\E_{Y\sim\calD_i}\left[(Y_i - \vp_i)^{\otimes t/2}\right]\left[(Y_i - \vp_i)^{\otimes t/2}\right]^{\top} \le 2N(8t/k)^{t/2}.
	\end{equation}
\end{proof}

\begin{lem}
	Suppose Conditions \ref{enum:satisfied}-\ref{enum:tensorconc} hold. Then for any $v\in\{\pm 1\}^n$, we have that \begin{equation}
		\text{\ref{program:basic}} \sos{O(t)} \left(\sum_{i\in S_g}(w_i - 1)\langle X_i - \vp_i, v\rangle\right)^t \le 2\epsilon^{t-1}\cdot N^t\cdot (8t/k)^{t/2}.\label{eq:B}
	\end{equation}
	\label{lem:B}
\end{lem}

\begin{proof}
	By SoS Holder's, we have that \begin{align}
		\text{\ref{program:basic}} \sos{O(t)}\left(\sum_{i\in S_g}(w_i - 1)\langle X_i - \vp_i, v\rangle\right)^t &= \left(\sum_{i\in S_g}(1 - w_i)\langle X_i - \vp_i, v\rangle\right)^t \\
		&\le \left(\sum_{i\in S_g}(1-w_i)\right)^{t-1}\left(\sum_{i\in S_g}\langle X_i - \vp_i, v\rangle^t\right) \\
		&\le (\epsilon N)^{t-1}\cdot \sum_{i\in[N]}\langle Y_i - \vp_i,v \rangle^t \\
		&\le (\epsilon N)^{t-1}\cdot 2N(8t/k)^{t/2} \\
	\end{align} where the third step follows from the fact that $\sos{2} \sum_{i\in S_g}(1 - w_i)\le \sum_{i\in[N]}(1-w_i) = \epsilon N$, and the fourth step follows from Lemma~\ref{lem:nesterov}.
\end{proof}

\begin{lem}
	Suppose Conditions \ref{enum:satisfied}-\ref{enum:tensorconc} hold. Then for any $v\in\{\pm 1\}^n$, we have that \begin{equation}
		\text{\ref{program:basic}} \sos{O(t)} \left(\sum_{i\in S_b}w_i \langle X_i - \hat{\vp}_i, v\rangle\right)^t \le 2\epsilon^{t-1}N^t(8t/k)^{t/2}.\label{eq:C}
	\end{equation}
	\label{lem:C}
\end{lem}

\begin{proof}
	We have that \begin{align}
		\text{\ref{program:basic}} \sos{O(t)} \left(\sum_{i\in S_b} w_i\langle X_i - \hat{\vp}_i, v\rangle\right)^t &= \left(\sum_{i\in S_b}w^2_i\langle X_i - \hat{\vp}_i,v\rangle\right)^t \label{eq:wi}\\
		&\le \left(\sum_{i\in S_b}w_i\right)^{t-1}\cdot\left(\sum_{i\in S_b}w_i\langle X_i - \hat{\vp}_i,v\rangle^t\right) \label{eq:holderonw}\\
		&\le \left(\sum_{i\in S_b}w_i\right)^{t-1}\cdot\left(\sum_{i\in [N]}w_i\langle X_i - \hat{\vp}_i,v\rangle^t\right) \label{eq:Sblessn}\\
		&\le |S_b|^{t-1}\cdot 2(8t/k)^{t/2}\sum_{i\in [N]}w_i \label{eq:usemomentbound}\\
		&= 2(\epsilon N)^{t-1}(8t/k)^{t/2}\cdot N\\
		&= 2\epsilon^{t-1}N^t(8t/k)^{t/2},
	\end{align} where \eqref{eq:wi} follows from the Booleanity constraints, \eqref{eq:holderonw} follows from SoS Holder's, \eqref{eq:Sblessn} follows from even-ness of $t$, \eqref{eq:usemomentbound} follows from the definition of $|S_b|$ and from the moment bound \eqref{eq:mainmomentbound}.
\end{proof}

We can now finish the proof of Lemma~\ref{lem:main}.

\begin{proof}[Proof of Lemma~\ref{lem:main}]
	By \eqref{eq:lemmasplit} and Lemmas~\ref{lem:A}, \ref{lem:B}, and \ref{lem:C}, we have that \begin{equation}
		\text{\ref{program:basic}}\sos{O(t)} \left(\sum_{i\in S_g}w_i\right)^t\langle\hat{\vp} - \vp,v\rangle^t \le O(N)^t\left(\delta^t + \epsilon^{t-1}(t/k)^{t/2}\right).
	\end{equation} Since $\text{\ref{program:basic}}\sos{2}\sum_{i\in S_g}w_i \ge (1 - 2\epsilon)N$, we conclude that \begin{equation}
		\text{\ref{program:basic}}\sos{O(t)} \langle\hat{\vp} - \vp,v\rangle^t \le O\left(\delta^t + \epsilon^{t-1}(t/k)^{t/2}\right)
	\end{equation} as claimed.
\end{proof}

\subsection{Rounding}
\label{subsec:round}

We are now ready to complete the proof of Theorem~\ref{thm:main_basic} by specifying how to round a pseudodistribution satisfying \ref{program:matrixsos_basic}.


\begin{lem}
	Let $\psE$ be a degree-$O(t)$ pseudodistribution satisfying $\text{\ref{program:matrixsos_basic}}$. Then $\psE[\hat{\vp}]\in\Delta^n$ and $\tvd(\psE[\hat{\vp}],\vp) \le O(\delta + \epsilon^{1-1/t}\sqrt{t/k})$. \label{lem:round}
\end{lem}

\begin{proof}
	The fact that $\psE[\hat{\vp}]$ follows from the fact that $\psE$ satisfies Constraints~\ref{constraint:simplex} of \ref{program:matrixsos_basic}. For the second part of the lemma, note that by the dual characterization of $L_1$ distance, it suffices to show that for any $v\in\{\pm 1\}^n$, \begin{equation}
		\langle \psE[\hat{\vp}] - \vp , v\rangle \le O\left(\delta + \epsilon^{1-1/t}\sqrt{t/k}\right)\label{eq:defnl1}
	\end{equation} By Lemma~\ref{lem:satisfied}, \ref{program:matrixsos_basic}$\sos{O(t)}$\ref{program:basic}.  Furthermore, by Lemma~\ref{lem:main}, for any $v\in\{\pm 1\}^n$, \begin{equation}\psE[\langle \hat{\vp} - \vp,v\rangle^t] \le O\left(\delta^t + \epsilon^{t-1}(t/k)^{t/2}\right),\end{equation} so we get that \begin{align}\langle \psE[\hat{\vp}] - \vp, v\rangle^t
	&\le \psE\left[\langle \hat{\vp} - \vp,v\rangle^t\right]\\
	&\le O\left(\delta^t + \epsilon^{t-1}(t/k)^{t/2}\right),\end{align} where the first step is a consequence of Fact~\ref{fact:scalarholders}. Now by the fact that $(a+b)^{1/t}\le a^{1/t} + b^{1/t}$ for positive scalars $a,b$, 
\end{proof}

We can now complete the proof of Theorem~\ref{thm:main_basic}.

\begin{proof}[Proof of Theorem~\ref{thm:main_basic}]
	The output of our algorithm will be $\textsc{Round}[\psE]$ for $\psE$ satisfying Program~\ref{program:matrixsos_basic} and therefore Program~\ref{program:basic}, so \textsc{Round} produces a hypothesis $h$ for which $\tvd(h,\vp)\le O(\delta + \epsilon^{1 - 1/t}\cdot\sqrt{t/k})$, as claimed.
\end{proof}


\section{Improved Sample Complexity Under Shape Constraints}
\label{sec:shape}

In this section we prove the following, which says that the algorithmic framework of the preceding sections can be leveraged to learn \emph{shape-constrained distributions} from untrusted batches with sample complexity \emph{sublinear} in the domain size $n$.

\begin{thm}
	Let $t\ge 4$ be any integer, and let $\eta > 0$. If $\vp$ is $(\eta,s)$-piecewise degree-$d$, then there is an algorithm that draws an $\epsilon$-corrupted set of $N$ $\delta$-diverse batches of size $k$ from $\vp$ for $N = \delta^{-2}\epsilon^{-2} (sd\log n)^{O(t)} \cdot k^t/t^{t-1}$, runs in time $\delta^{-t}\epsilon^{-t}(sdn)^{O(t)}\cdot k^{t^2}/t^{t(t-1)}$, and with probability $1 - 1/\poly(n)$ outputs a distribution $\hat{\vp}$ for which $\tvd(p,\hat{\vp}) \le O(\eta + \delta + \epsilon^{1-1/t}\sqrt{t/k})$.
	\label{thm:main_shape}
\end{thm}
\noindent
Importantly, by combining this result with known approximation theoretic results, we are able to obtain sample complexities that are either independent of the domain size or depend at most polylogarithmically on it, for a large class of natrual distributions, such as monotone distributions, monotone hazard rate distributions, log-concave distributions, discrete Guassians, Poisson Binomial distributions, and mixtures thereof, see e.g.~\cite{acharya2017sample} for more details.
After giving the basic ingredients from VC complexity for how to learn shape-constrained distributions in sublinear sample complexity in a classical sense, we describe and analyze the polynomial system Program~\ref{program:shape_basic}, deferring technical details for how to encode some of the constraints of this program to Section~\ref{sec:matrixsos} and Appendix~\ref{app:defer}.

\subsection{\texorpdfstring{$\calA_K$}{AK} Norms and VC Complexity}
\label{subsec:AK_VC}

\begin{defn}[$\calA_K$ norms, see e.g.~\cite{devroye2001combinatorial}]\label{def:ak}
	For positive integers $K\le n$, define $\calA_K$ to be the set of all unions of at most $K$ disjoint intervals over $[n]$, where an interval is any subset of $[n]$ of the form $\{a,a+1,...,b-1,b\}$. The $\A{K}$ distance between two distributions $p,q$ over $[n]$ is \begin{equation}
		\norm{p-q}_{\A{K}} = \max_{S\in\A{K}}|p(S) - q(S)|.
	\end{equation} Equivalently, say that $v\in\{\pm 1\}^n$ has \emph{$2K$ sign changes} if there are exactly $2K$ indices $i\in[n-1]$ for which $v_{i+1}\neq v_i$. Then if $\mathcal{V}^n_{2K}$ denotes the set of all such $v$, we have \begin{equation}
		\norm{p-q}_{\A{K}} = \frac{1}{2}\max_{v\in\mathcal{V}^n_{2K}}\langle p - q, v\rangle.
	\end{equation} Note that \begin{equation}
		\norm{\cdot}_{\A{1}} \le \norm{\cdot}_{\A{2}}\le \cdots \le \norm{\cdot}_{\A{n/2}} = \norm{\cdot}_{\text{TV}}.
	\end{equation}
\end{defn}


\begin{defn}
	We say that a distribution over $[n]$ is \emph{$(\eta,s)$-piecewise degree-$d$} if there is a partition of $[n]$ into $t$ disjoint intervals $\{[a_i,b_i]\}_{1\le i\le t}$, together with univariate degree-$d$ polynomials $r_1,...,r_t$ and a distribution $\vec{q}$ on $[n]$, such that $\tvd(\vec{p},\vec{q})\le \eta$ and such that for all $i\in[t]$, $\vec{q}(x) = r_i(x)$ for all $x\in[n]$ in $[a_i,b_i]$.
\end{defn}

\begin{lem}
	Let $K = s(d+1)$. If $\vp$ is $(\eta,s)$-piecewise degree-$d$ and $\norm{\vp - \hat{\vp}}_{\calA_K} \le \zeta$, then there is an algorithm which, given the vector $\hat{\vp}$, outputs a distribution $\vp^*$ for which $\tvd(\vp,\vp^*)\le 2\zeta + 4\eta$ in time $\poly(s,d,1/\epsilon)$.\label{lem:piecewise_vc}
\end{lem}

\begin{proof}
	Let $\vq$ be a $(0,s)$-piecewise degree-$d$ distribution for which $\tvd(\vp,\vq) = \eta$. By Theorem~\ref{thm:piecewise} below, one can produce an $s$-piecewise degree-$d$ distribution $\vp^*$ minimizing $\norm{\hat{\vp} - \vp^*}_{\calA_K}$ to within additive error $\eta$ in time $\poly(s,d,1/\eta)$. We already know by triangle inequality that \begin{equation}
		\norm{\hat{vp} - \vq}_{\calA_K} \le \norm{\hat{vp} - \vp}_{\calA_K} + \norm{\vp - \vq}_{\calA_K} \le \zeta + \eta,
	\end{equation} so by $\eta$-approximate minimality we know $\norm{\hat{\vp} - \vp^*}_{\calA_K} \le \zeta + 2\eta$. By another application of triangle inequality, we conclude that $\norm{\vq - \vp^*}_{\calA_K} \le 2\zeta + 3\eta$. Because $\vq$ and $\vp^*$ are both $s$-piecewise degree-$d$, the vector $\vq - \vp^*$ has at most $2s(d+1)$ sign changes. Indeed, the common refinement of the intervals defining the two piecewise polynomials is at most $2s$ intervals, and the difference between two degree-$d$ polynomials over any of these intervals is degree-$d$ (the additional $+1$ comes from the endpoints of each of the intervals). So we get that $\tvd(\hat{\vp} - \vp^* = \norm{\hat{\vp} - \vp^*}_{\calA_K} \le 2\zeta + 3\eta$, and one final application of triangle inequality allows us to conclude that $\tvd(\vp - \vp^*) = 2\zeta + 4\eta$.
\end{proof}

\begin{thm}[\cite{acharya2017sample}]
	There is an algorithm which, given a vector $\vp\in\Delta^n$, computes an $s$-piecewise degree-$d$ hypothesis $h$ which minimizes $\norm{h - \vp}_{\calA_{s(d+1)}}$ to within additive error $\gamma$ in time $n \cdot \poly(d,1/\gamma)$.\label{thm:piecewise}
\end{thm}

\subsection{Another SoS Relaxation}

Henceforth, let $K = s\cdot(d+1)$ and let $\ell = 2s(d+1)$. To prove Theorem~\ref{thm:main_shape}, by Lemma~\ref{lem:piecewise_vc} it suffices to learn $\vp$ in $\calA$ distance, that is, we wish to produce a hypothesis $\hat{\vp}$ for which $\frac{1}{2}\max_{v\in\mathcal{V}^n_{\ell}} \langle\vp - \hat{\vp},v\rangle$ is small.

\begin{program}{$\calP'$}
	The variables are $\{w_i\}_{i\in[N]}$, $\{\hat{\vp}_i\}_{i\in[N]}$, and $\hat{\vp}$, and the constraints are

	\begin{enumerate}
	 	\item $w^2_i = w_i$ for all $i\in[N]$.
	 	\item $\sum w_i = (1 - \epsilon)N$.
	 	\item For every $v\in\{\pm 1\}^n$ with at most $\ell$ sign changes and every $i\in[N]$, $\langle \hat{\vp}_i - \hat{\vp}, v\rangle \le 5\delta$.\label{constraint:l1_shape}
	 	\item $\sum_{i\in[N]}w_i X_i = \hat{\vp}\sum_{i\in[N]}w_i$.
	 	\item For every $v\in\{\pm 1\}^n$ with at most $\ell$ sign changes, \begin{equation}
		\sum_{i\in[N]}w_i\langle X_i - \hat{\vp}_i,v\rangle^t \le (8t/k)^{t/2}\cdot \sum_{i\in[N]}w_i.\label{eq:mainmomentbound_shape}
		\end{equation}\label{constraint:momentboundprime}
		\item $\hat{\vp}_i \ge 0$ for all $i\in[n]$ and $\sum_i \hat{\vp}_i = 1$. \label{constraint:simplex_prime}
	\end{enumerate}
	\label{program:shape_basic}
\end{program}


\begin{lem}
There is a system \ref{program:matrixsos_shape} of degree-$O(t)$ polynomial equations and inequalities in the variables $\{w_i\}$, $\{\hat{\vp}_i\}$, $\hat{\vp}$, and $n^{O(t)}$ other variables, whose coefficients depend on $\epsilon, t, X_1,...,X_n$ such that

\begin{enumerate}
	\item (Satisfiability) With probability at least $1- 1/\poly(n)$, \ref{program:matrixsos_shape} has a solution in which $\hat{\vp} = \vp$ and for each $i\in[N]$, $\hat{\vp}_i = \vp_i$ and $w_i$ is the indicator for whether $X_i$ is an uncorrupted point.
	\item (Encodes Moment Bounds) $\text{\ref{program:matrixsos_shape}} \sos{O(t)}\text{\ref{program:shape_basic}}$.
	\item (Solvability) If \ref{program:matrixsos_shape} is satisfied, then for every integer $C > 0$, there is an $n^{O(Ct)}$-time algorithm which outputs a degree-$Ct$ pseudodistribution which satisfies \ref{program:matrixsos_basic} up to additive error $2^{-n}$.
\end{enumerate}
\label{lem:satisfied_shape}
\end{lem}

Together with Lemma~\ref{lem:piecewise_vc}, this suggests the following algorithm for learning from untrusted batches when $\vp$ is $(\eta,s)$-piecewise degree-$d$: use semidefinite programming to efficiently obtain a pseudodistribution over solutions to Program~\ref{program:matrixsos_basic}, round this pseudodistribution to an estimate for $\vp$ by computing the pseudoexpectation of the $\hat{\vp}$ variable, and then refine this by computing the best piecewise polynomial approximation to this estimate. The only difference between this algorithm and \textsc{LearnFromUntrusted} is the the third step.

A formal specification of this algorithm, which we call \textsc{PiecewiseLearn}, is given in Algorithm~\ref{alg:shape} below.

\begin{center}
	\myalg{alg:shape}{PiecewiseLearn}{
		\textbf{Input}: Corruption parameter $\epsilon$, diversity parameter $\delta$, support size $n$, batch size $k$, samples $\{X_i\}_{i\in[N]}$, degree $t$, $(\eta,s,d)$ for which $\vp$ is $(\eta,s)$-piecewise degree-$d$ \\
		\textbf{Output}: Estimate $\vp^*$
		\begin{enumerate}
			\item Run SDP solver to find a pseudodistribution $\psE$ of degree $O(t)$ satisfying the constraints of Program~\ref{program:matrixsos_basic}.
			\item Set $\tilde{\vp}\triangleq \psE[\hat{\vp}]$.
			\item Let $K = s(d+1)$. Using the algorithm of \cite{acharya2017sample}, output the $s$-piecewise degree-$d$ distribution $\vp^*$ that minimizes $\norm{\tilde{\vp} - \vp^*}_{\calA_K}$ (up to additive error $\eta$)
		\end{enumerate}
	}
\end{center}

\subsection{Deterministic Conditions and Identifiability}

We will condition on the following deterministic conditions holding simultaneously: \begin{enumerate}[label=(\Roman*)]
	\item \label{enum:satisfied_shape} The ``Satisfiability'' condition of Lemma~\ref{lem:satisfied_shape} holds.
	\item \label{enum:meanconc_shape} The mean of the uncorrupted points concentrates in $\A{\ell}$ norm: \begin{equation}
		\norm{\frac{1}{N}\sum_{i\in S_g}(X_i - \vp_i)}_{\A{\ell}}\le O(\epsilon^{1-1/t}\sqrt{t/k}) \label{eq:uncorruptconc_shape}
	\end{equation}
	\item \label{enum:tensorconc_shape} For every $v\in\{\pm 1\}^n$ with at most $\ell$ sign changes, \begin{equation}
		\left|\frac{1}{N}\sum_{i\in[N]}\langle Y_i - \vp_i,v\rangle^t - \frac{1}{N}\sum_{i\in[N]}\E_{Y_i\sim \calD_i}\langle Y_i - \vp_i,v\rangle^t\right| \le (8t/k)^{t/2}
	\end{equation}
\end{enumerate}

\begin{lem}
	Conditions \ref{enum:satisfied_shape}, \ref{enum:meanconc_shape}, \ref{enum:tensorconc_shape} hold simultaneously with probability $1 - 1/\poly(n)$.\label{lem:conc_shape}
\end{lem}

\begin{proof}
	(I) holds with probability at least $1 - 1/\poly(n)$ according to Lemma~\ref{lem:satisfied_shape}.

	For (II), we will apply Lemma~\ref{lem:shell_conc} with $\mathcal{N}$ taken to be the collection of all $v\in\{\pm 1\}^n$ with at most $\ell$ sign changes. $|\mathcal{N}| = n^{O(\ell)}$, so provided $|S_g| \ge \Omega((k/t)\epsilon^{-2}\cdot \ell\log^2 n)$, we get that (II) holds with probability $1 - 1/\poly(n)$.

	For (III), we will apply Lemma~\ref{lem:shell_conc_shape} with $\mathcal{N}$ taken to be the collection of all $v^{\otimes t/2}(v^{\otimes t/2})^{\top}$ for which $v\in\{\pm 1\}^n$ has at most $\ell$ sign \\
	ges. $|\mathcal{N}| = n^{O(\ell)}$, so when $N\ge \Omega((k/8t)^t\cdot \ell\log n)$, (III) holds with probability $1 - 1/\poly(n)$.
\end{proof}

The SoS proof of identifiability given Program~\ref{program:shape_basic} is identical to the proof of identifiability given Program~\ref{program:basic} in Section~\ref{subsec:basic_id}, the only difference being that all intermediate steps in the proof are quantified over $v\in\{\pm 1\}^n$ with at most $\ell$ sign changes, rather than over all $v\in\{\pm 1\}^n$. This yields the following:

\begin{lem}
	Suppose Conditions \ref{enum:satisfied_shape}-\ref{enum:tensorconc_shape} hold. Then for any $v\in\{\pm 1\}^n$ with at most $\ell$ sign changes, we have that \begin{equation}
		\text{\ref{program:shape_basic}} \sos{O(t)} \langle \hat{\vp} - \vp, v\rangle^t \le O\left(\delta^t + \epsilon^{t-1}(t/k)^{t/2}\right).
	\end{equation}\label{lem:main_shape}
\end{lem}

\subsection{Rounding}

Once we have Lemma~\ref{lem:main_shape}, the rounding step can be analyzed in essentially the same way as Lemma~\ref{lem:round}. We include a proof for completeness.

\begin{lem}
	Let $\psE$ be a pseudoexpectation satisfying Program~\ref{program:matrixsos_shape}. Then $\psE[\hat{\vp}]\in\Delta^n$ and $\norm{\psE[\hat{\vp}] - \vp}_{\calA_{\ell}} \le O(\delta + \epsilon^{1 - 1/t}\cdot\sqrt{t/k})$.
\end{lem}

\begin{proof}
	$\psE[\hat{\vp}]\in\Delta^n$ because $\psE$ satisfies Constraint~\ref{constraint:simplex_prime} of \ref{program:shape_basic}. For the second part of the lemma, by definition of $\calA_{\ell}$ distance, it suffices to show that for any $v\in\{\pm 1\}^n$ with at most $\ell$ sign changes, \begin{equation}
		\langle \psE[\hat{\vp}] - \vp,v\rangle \le O\left(\delta + \epsilon^{1-1/t}\sqrt{t/k}\right).
	\end{equation} By Lemma~\ref{lem:satisfied_shape}, \ref{program:matrixsos_shape} $\sos{O(t)}$ \ref{program:shape_basic}. Furthermore, by Lemma~\ref{lem:main_shape}, for any $v\in\{\pm 1\}^n$ with at most $\ell$ sign changes, \begin{equation}
		\psE[\langle \hat{\vp} - \vp, v\rangle^t] \le O\left(\delta^t + \epsilon^{t-1}(t/k)^{t/2}\right),
	\end{equation} so we get that \begin{align*}
		\langle \psE[\hat{\vp}] - \vp, v\rangle^t &\le \psE[\langle \hat{\vp} - \vp, v\rangle^t] \\
		&\le O\left(\delta^t + \epsilon^{t-1}(t/k)^{t/2}\right),
	\end{align*} where the first step is a consequence of Fact~\ref{fact:scalarholders}. Now by the fact that $(a + b)^{1/t} \le a^{1/t} + b^{1/t}$ for positive scalars $a,b$, the lemma follows.
\end{proof}

We can now complete the proof of Theorem~\ref{thm:main_shape}.

\begin{proof}[Proof of Theorem~\ref{thm:main_shape}]
	The output of our algorithm will be $\textsc{Round}[\psE]$ for $\psE$ satisfying Program~\ref{program:shape_basic}. Because $\psE[\hat{\vp}]\in\Delta^n$ satisfies $\norm{\psE[\hat{\vp}] - \vp}_{\calA_{\ell}} \le O(\delta + \epsilon^{1 - 1/t}\cdot\sqrt{t/k})$, we conclude that by Lemma~\ref{lem:piecewise_vc}, the assumption that $\vp$ is $(\eta,s)$-piecewise degree-$d$, and the fact that $\ell = 2s(d+1)$, \textsc{Round} produces a hypothesis $h$ for which $\tvd(h,\vp)\le O(\eta + \delta + \epsilon^{1 - 1/t}\cdot\sqrt{t/k})$, as claimed.
\end{proof}


\section{Encoding Moment Constraints}
\label{sec:matrixsos}

In this section we will prove Lemmas~\ref{lem:satisfied} and \ref{lem:satisfied_shape}. The programs \ref{program:matrixsos_basic} and \ref{program:matrixsos_shape} referenced in those Lemmas will involve systems of inequalities among matrix-valued polynomials. We begin by giving an overview of how such inequalities fit into the SoS proof system.

\subsection{Matrix SoS Proofs}
\label{subsec:matrixsosproofs}


Let $x_1,...,x_n$ be formal variables. In this subsection we show how the SoS proof system can reason about constraints of the form $M(x)\succeq 0$, where $M(x)$ is some symmetric matrix whose entries are polynomials in $x$.

Let $M_1(x),...,M_m(x)$ be symmetric matrix-valued polynomials of $x$ of various sizes ($1\times 1$ matrix-valued polynomials are simply scalar polynomials), and let $q_1(x),...,q_{m}(x)$ be scalar polynomials. The expression \begin{equation}
	\{M_1\succeq 0,...,M_m\succeq 0, q_1(x) = 0,...,q_m(x)=0\} \sos{d} p(x)\ge 0
\end{equation} means that there exists a vector $u$, a matrix $Q(x)$ whose entries are polynomials in the ideal generated by $q_1,...,q_m$, and vector-valued polynomials $\{r^j_S\}_{j\le N,S\subseteq[m]}$ (where $S$'s are multisets) for which \begin{equation}
	p(x) = Q(x) + u^{\top}\left[\sum_{S\subseteq[m]}\left(\sum_j(r^j_S(x))(r^j_S(x))^{\top}\right)\otimes \left[\otimes_{i\in S}M_i(x)\right]\right]u\label{eq:matrixsos}
\end{equation} and $Q(x)$ and the entries of each summand in \eqref{eq:matrixsos} are all polynomials of degree at most $d$.

A pseudodistribution $\psE$ of degree $2d$ is said to satisfy $\{M_1(x)\succeq 0,...,M_m(x)\succeq 0\}$ if for every multiset $S\subseteq[m]$ and polynomial $p(x)$ for which the entries of $p(x)^2 \cdot (\otimes_{i\in S}M_i(x))$ are degree at most $2d$, we have \begin{equation}
	\psE[p(x)^2 \cdot (\otimes_{i\in S}M_i(x))] \succeq 0.
\end{equation} Such pseudodistributions can still be found efficiently via semidefinite programming.

Proofs of the following basic lemmas about matrix SoS can be found in \cite{hopkins2018mixture}.

\begin{lem}[\cite{hopkins2018mixture}, Lemma 7.1]
	If $\psE$ is a degree-$2d$ pseudodistribution satisfying $\{M_1\succeq 0,...,M_m\succeq 0\}$ and furthermore \begin{equation}
		\{M_1\succeq 0,...,M_m\succeq 0\} \sos{2d} M\succeq 0,
	\end{equation} then $\psE$ also satisfies $\{M_1\succeq 0,...,M_m\succeq 0,M\succeq 0\}$.
\end{lem}

\begin{lem}[\cite{hopkins2018mixture}, Lemma 7.2]
	If $f(x)$ is a degree-$d$ vector-valued polynomial of dimension $s$ and $M(x)$ is an $s\times s$ symmetric matrix-valued polynomial of degree $d'$, then \begin{equation}
		\{M\succeq 0\}\sos{dd'} \langle f(x),M(x)f(x)\rangle \ge 0.
	\end{equation}\label{lem:matrixsosimply}
\end{lem}

\subsection{Moment Constraints for Program~\ref{program:basic}}
\label{subsec:encode_basic}

We first show how to encode Constraint~\ref{constraint:l1} of Program~\ref{program:basic}, namely that for each $i\in[N]$ \begin{equation}
	\{v^2_i= 1 \ \forall \ 1\le i\le n\} \sos{2} \langle\hat{\vp}_i - \hat{\vp},v\rangle \le 5\delta.\label{eq:sosproofl1}
\end{equation} This would hold if there existed sum-of-squares polynomials $q_S(v,\hat{p}_i,\hat{p})$ for which $5\delta - \langle\hat{\vp}_i - \hat{\vp},v\rangle = \sum_S\prod_{i\in S}(1-v^2_i)\cdot q_S(v,\hat{p}_i,\hat{p})$ such that each summand on the right-hand side is of degree at most 2. So let $Q^S$ be an $n\times n$ matrix of indeterminates, with entries indexed by $i,j\in[n]$, which will correspond to the matrix of coefficients of $q(v,\hat{p}_i,\hat{p})$ as a quadratic polynomial in $v$.

Next we show how to encode Constraint~\ref{eq:mainmomentbound} of Program~\ref{program:basic}. For every $S\subset[n]$ of size at most $O(t)$, let $M^S$ be an $n^{t/2}\times n^{t/2}$ matrix of indeterminates, one for each pair of multi-indices $\gamma,\rho$ over $[n]$ both of degree at most $t/2$. We would like to impose constraints on the entries $M^S_{\gamma,\rho}$ so that psd-ness of the matrices in $\{M^S: S\subseteq[n]\}$ encodes the fact that \begin{equation}\{v^2_i= 1 \ \forall 1\le i\le n\}\sos{O(t)} \sum_{i\in[N]}w_i\langle X_i - \hat{\vp}_i,v\rangle^t \le 2\cdot (8t/k)^{t/2}\sum_{i\in[N]}w_i\label{eq:hassos}\end{equation} Recall that the condition \eqref{eq:hassos} means that there exist polynomials $p_S$ for which \begin{equation}
	2\cdot (8t/k)^{t/2}\sum_{i\in[N]}w_i - \sum_{i\in[N]}w_i\langle X_i - \hat{\vp}_i,v\rangle^t = \sum_{S: |S|\le O(t)}p_S(v,\{w_i\},\{\hat{\vp}_i\},\hat{\vp})\cdot\prod_{i\in S}(1-v_i^2),
\end{equation} where each $p_S$ is a sum-of-squares polynomial such that $p_S(v,\{w_i\},\{\hat{\vp}_i\},\hat{\vp})\cdot\prod_{i\in S}(1 - v^2_i)$ is degree $O(t)$. $M^S$ will correspond to the matrix of coefficients of $p_S(v,\{w_i\},\{\hat{\vp}_i\},\hat{\vp})$ as a degree-$t$ polynomial in $v$. Specifically, we will consider the following program.

\begin{program}{$\hat{\calP}$}
	The variables are $\{w_i\}_{i\in[N]}$, $\hat{\vp}$, $\{\hat{\vp}_i\}_{i\in[N]}$, $\{Q^S_{i,j}\}$, and $\{M^S_{\gamma,\rho}\}$ and the constraints are \begin{enumerate}
	\item $w^2_i = w_i$ for all $i\in[N]$.\label{constraint:boolean}
	\item $\sum w_i = (1-\epsilon)N$.\label{constraint:cardinality}
	\item $5\delta - \langle\hat{\vp}_i - \hat{\vp},v\rangle = \sum_S\prod_{i\in S}(1 - v_i^2)\cdot \langle v,Q^Sv\rangle$\label{constraint:l1sos}
	\item $\sum_{i\in[N]}w_i X_i = \hat{\vp}\cdot \sum_{i\in[N]}w_i$\label{constraint:mean}
	\item \begin{equation}2\cdot (8t/k)^{t/2} - \frac{1}{(1-\epsilon)N}\sum_{i\in[N]}w_i\langle X_i - \hat{\vp},v\rangle^t = \sum_{S: |S|\le O(t)}\prod_{i\in S}(1-v_i^2)\cdot\langle v^{\otimes t/2}, M^S v^{\otimes t/2}\rangle\end{equation}\label{constraint:momentboundsos}
	\item $Q^S\succeq 0$ for all $S\subset[n]$ for which $|S|\le 2$\label{constraint:Qpsd}
	\item $M^S\succeq 0$ for all $S\subset[n]$ for which $|S|\le O(t)$. \label{constraint:Mpsd}
	\item $\hat{\vp}_i \ge 0$ for all $i\in[n]$ and $\sum_i \hat{\vp}_i = 1$. \label{constraint:simplex_matrixsos_basic}
\end{enumerate}\label{program:matrixsos_basic}
\end{program}

\begin{defn}
	Define the \emph{canonical assignment} to the variables $\{w_i\}_{i\in[N]}$, $\hat{\vp}$, and $\{\hat{\vp}_i\}_{i\in[N]}$ to be as follows: for each $i\in[N]$, $w_i = \bone[X_i \ \text{is uncorrupted}]$, $\hat{\vp}_i = \vp_i$, and $\hat{\vp} = \frac{1}{(1-\epsilon)N}\sum_i w_i X_i$.
\end{defn}

\begin{proof}[Proof of Lemma~\ref{lem:satisfied}]
	The fact that \ref{program:matrixsos_basic} $\sos{O(t)}$ \ref{program:basic} follows by Lemma~\ref{lem:matrixsosimply}, and solvability follows from the fact that the problem of outputting a degree-$O(t)$ pseudodistribution satisfying a system of degree-$O(t)$ polyomial constraints can be encoded as a semidefinite program of size $n^{O(t)}$.

	It remains to show satisfiability of Program~\ref{program:matrixsos_basic}. Constraints~\ref{constraint:boolean}, \ref{constraint:cardinality}, and \ref{constraint:mean} are clearly satisfied by the canonical assignment.

	For Constraints~\ref{constraint:l1sos} and \ref{constraint:Qpsd}, we want to show that for each $i\in[N]$, the SoS proof \eqref{eq:sosproofl1} exists as a polynomial inequality only in the variable $v$, with $\{\hat{\vp}_i\}$ and $\hat{\vp}$ now fixed. Fix any $i\in[N]$ and for convenience define $\alpha_j = (\hat{\vp}_i - \hat{\vp})_j$. From Fact~\ref{fact:1bounded}, we get that \begin{equation}
		\{v^2_i= 1 \ \forall \ 1\le i\le n\} \sos{2} \langle\hat{\vp}_i - \hat{\vp},v\rangle = \sum^n_{j=1}\alpha_j v_j \le \sum^n_{j=1}|\alpha_j| = \norm{\hat{\vp}_i - \hat{\vp}}_1
	\end{equation}. 

	By triangle inequality and the fact that $\tvd(\vp_i,\vp_j) \le 2\delta$ for all $j\in[N]$, \begin{align*}
		\norm{\hat{\vp}_i - \hat{\vp}}_1 &\le \norm{\frac{1}{(1-\epsilon)N}\sum_{j\in S_g: j\neq i}(\vp_i - \vp_j)}_1 + \norm{\frac{1}{(1-\epsilon)N}\sum_{j\in S_g}(X_j - \vp_j)}_1 \\
		&\le 4\delta + \norm{\frac{1}{(1-\epsilon)N}\sum_{j\in S_g}(X_j - \vp_j)}_1
	\end{align*} By Fact~\ref{fact:multconc} and the fact that $\{X_j\}_{j\in S_g}$ is a collection of independent draws from $\{\mul_k(\vp_j)\}_{j\in S_g}$ respectively, we know that \begin{equation}
		\norm{\frac{1}{(1-\epsilon)N}\sum_{j\in S_g}(X_j - \vp_j)}_1 \le \delta
	\end{equation} with probability at least $1 - n\cdot e^{-2\delta^2 N/n^2}$, from which \eqref{eq:sosproofl1} follows.

	Finally, for Constraints~\ref{constraint:momentboundsos} and \ref{constraint:Mpsd}, suppose the following SoS proof exists: \begin{equation}\{v^2_i= 1 \ \forall \ 1\le i\le n\}\sos{O(t)} \frac{1}{(1-\epsilon)N}\sum_{i\in S_g}\langle X_i - \hat{\vp}_i,v\rangle^t \le 2\cdot(8t/k)^{t/2}\label{eq:suppose},\end{equation} where $v$ is the only variable and $\{w_i\}$, $\{\hat{\vp}_i\}$, and $\vp$ have all been fixed. By definition, this means that there exist sum-of-squares polynomials $p_S(v)$ for every $S\subset[n]$ of size at most $O(t)$ such that $p_S(v)\cdot\prod_{i\in S}(1 - v_i^2)$ is degree $O(t)$ and \begin{equation}
		2\cdot (8t/k)^{t/2} - \frac{1}{(1-\epsilon)N}\sum_{i\in S_g}w_i\langle X_i - \hat{\vp}_i,v\rangle^t = \sum_{S: |S|\le O(t)}p_S(v)\cdot\prod_{i\in S}(1-v_i^2).\end{equation} By taking $M^S$ to be the matrix of coefficients for which $\langle v^{\otimes t/2},M^Sv^{\otimes t/2}\rangle = p_S(v)$ and noting that $M^S\succeq 0$ because $p_S$ is an SoS, we satisfy the remaining Constraints~\ref{constraint:momentboundsos} and \ref{constraint:Mpsd} of \ref{program:matrixsos_basic}.

	It remains to verify that the SoS proof \eqref{eq:suppose} exists with high probability. Because $\hat{\vp}_i = \vp_i$, it is enough to show that the SoS proof \begin{equation}\{v^2_i= 1 \ \forall 1\le i\le n\}\sos{O(t)} \frac{1}{(1-\epsilon)N}\sum_{i\in S_g}\langle X_i - \vp_i,v\rangle^t \le 2\cdot (8t/k)^{t/2}\label{eq:suppose2},\end{equation} exists. It is enough to bound the quantity \begin{equation} b(v) \triangleq \frac{1}{(1-\epsilon)N}\sum_{i\in S_g}\langle X_i - \vp_i, v\rangle^t - \frac{1}{(1-\epsilon)N}\sum_{i\in S_g}\E_{X\sim\calD_i}\langle X - \vp_i,v\rangle^t\end{equation} by $b(v)\le (8t/k)^{t/2}$. Together with Lemma~\ref{lem:basicbound}, this will conclude the proof. But the desired bound on $b(v)$ follows by condition~\ref{enum:tensorconc} in Lemma~\ref{lem:conc}, with probability $1 - 1/\poly(n)$.
\end{proof}

\subsection{Moment Constraints for Program~\ref{program:shape_basic}}
\label{encode:shape}

The only changes in going from Program~\ref{program:basic} to Program~\ref{program:shape_basic} are Constraints~\ref{constraint:l1} and \ref{constraint:momentbound}. In this section, we explain how to succinctly quantify over all $v\in\{\pm 1\}^n$ with at most $\ell$ sign changes. To describe this encoding, we first recall some basic facts about the (discretized) Haar wavelet basis.

\paragraph{Haar Wavelets}

\begin{defn}
	Let $m$ be a positive integer and let $n = 2^m$. the \emph{Haar wavelet basis} is an orthonormal basis over $\R^n$ consisting of the \emph{father wavelet} $\psi_{0_{\father},0} = n^{-1/2}\cdot \ones$, the \emph{mother wavelet} $\psi_{0_{\mother},0} = n^{-1/2}\cdot(1,...,1,-1,...,-1)$ (where $(1,...,1,-1,...,-1)$ contains $n/2$ 1's and $n/2$ -1's), and for every $i,j$ for which $1\le i < m$ and $0\le j< 2^{i}$, the wavelet $\psi_{i,j}$ whose $2^{m - i}\cdot j + 1,...,2^{m-i}\cdot j + 2^{m-i-1}$-th coordinates are $2^{-(m-i)/2}$ and whose $2^{m-i}\cdot j + (2^{m-i-1} + 1),...,2^{m-i}\cdot j + 2^{m-i}$-th coordinates are $-2^{-(m-i)/2}$, and whose remaining coordinates are 0.

	Let $H_m$ denote the $n\times n$ matrix whose rows consist of the vectors of the Haar wavelet basis for $\R^n$. When the context is clear, we will omit the subscript and refer to this matrix as $H$.
\end{defn}

\begin{example}
	The Haar wavelet basis for $\R^8$ consists of the vectors \begin{align}
		\psi_{0_{\father},0} & = 2^{-3/2}(1,1,1,1,1,1,1,1) \\
		\psi_{0_{\mother},0} & = 2^{-3/2}(1,1,1,1,-1,-1,-1,-1) \\
		\psi_{1,0} &= 2^{-1}(1,1,-1,-1,0,0,0,0) \\
		\psi_{1,1} &= 2^{-1}(0,0,0,0,1,1,-1,-1) \\
		\psi_{2,0} &= 2^{-1/2}(1,-1,0,0,0,0,0,0) \\
		\psi_{2,1} &= 2^{-1/2}(0,0,1,-1,0,0,0,0) \\
		\psi_{2,2} &= 2^{-1/2}(0,0,0,0,1,-1,0,0) \\
		\psi_{2,3} &= 2^{-1/2}(0,0,0,0,0,0,1,-1)
	\end{align}
\end{example}

The key observation is that there is an orthonormal basis under which any $v\in\{\pm 1\}^n$ with at most $\ell$ sign changes has an $(\ell\log n + 1)$-sparse representation.

Define $\mathcal{T} \triangleq \{0_{\father}, 0_{\mother}, 1,...,m-1\}$. By abuse of notation, we will sometimes identify the indices $0_{\father}$ and $0_{\mother}$ with their numerical value of 0.

\begin{lem}
	Let $v\in\{\pm 1\}^n$ have at most $\ell$ sign changes. Then \begin{equation}
		\sum_{i\in\mathcal{T}}\sum^{2^i-1}_{j=0}2^{-(m-i)/2}|\langle \psi_{i,j}, v\rangle| \le \ell\log n + 1.\label{eq:main_haar_inequality}
	\end{equation}.\label{lem:haar}
\end{lem}

\begin{proof}
	We first show that $Hv$ has at most $\ell\log n + 1$ nonzero entries. For any $\psi_{i,j}$ with nonzero entries at indices $[a,b]\subset[n]$ and such that $i\neq 0_{\father}$, if $v$ has no sign change in the interval $[a,b]$, then $\langle \psi_{i,j},v\rangle = 0$. For every index $\nu\in[n]$ at which $v$ has a sign change, there are at most $m = \log n$ choices of $i,j$ for which $\psi_{i,j}$ has a nonzero entry at index $\nu$, from which the claim follows by a union bound over all $\ell$ choices of $\nu$, together with the fact that $\langle \psi_{0_{\father},0},v\rangle$ may be nonzero.

	Now for each $(i,j)$ for which $\langle\psi_{i,j},v\rangle\neq 0$, note that \begin{equation}2^{-(m-i)/2}\cdot |\langle\psi_{i,j},v\rangle| \le 2^{-(m-i)/2}\cdot\left(2^{-(m-i)/2}\cdot 2^{m-i}\right) = 1,\end{equation} from which \eqref{eq:main_haar_inequality} follows.
\end{proof}

For notational simplicity in the arguments below, for $\nu\in[n]$, if the $\nu$-th element of the Haar wavelet basis for $\R^n$ is some $\psi_{i,j}$, then let $\mu^{(\nu)}$ denote the weight $2^{-(m-i)/2}$. Also, for any $i\in\mathcal{T}$, let $T_i\subset[n]$ denote the set of all indices $\nu$ for which the $\nu$-th Haar wavelet is of the form $\psi_{i,j}$ for some $j$.

\paragraph{The Matrix SoS Encoding}

By Lemma~\ref{lem:haar}, instead of quantifying over all $v\in\{\pm 1\}^n$ with at most $\ell$ sign changes in Constraints~\ref{constraint:l1_shape} and \ref{constraint:momentboundprime} of \ref{program:shape_basic}, we can quantify over all $v\in\R^n$ with Frobenius norm at most $n$ and for which \eqref{eq:main_haar_inequality} is satisfied. Specifically, we can ask for an SoS proof of \begin{equation}\langle \hat{\vp}_i - \hat{\vp},v\rangle\le 5\delta\label{eq:relaxedl1}\end{equation} using Axioms~\ref{axioms:l1_shape}.

\begin{axioms}[Axioms for Constraint~\ref{constraint:l1_shape}]\label{axioms:l1_shape}
Let $\vec{W}_1,...,\vec{W}_n$ be auxiliary scalar variables.
	\begin{enumerate}
	\item $v_i^2 = 1$ for all $i\in[n]$
	\item $-\vec{W}_i \le (Hv)_i \le \vec{W}_i$ for all $i\in[n]$
	\item $\sum_i \mu^{(i)}\cdot \vec{W}_i \le\ell\log n + 1$,
\end{enumerate}
\end{axioms}

Likewise, we can ask for an SoS proof of \begin{equation}
	\frac{1}{(1-\epsilon)N}\sum_{i\in[N]}w_i \langle X_i - \hat{\vp}_i, v\rangle^t \le 2\cdot (8t/k)^{t/2},\label{eq:relaxedmoment}
\end{equation} using Axioms~\ref{axioms:momentboundprime}.

\begin{axioms}[Axioms for Constraint~\ref{constraint:momentboundprime}]\label{axioms:momentboundprime} Let $\{\vec{U}_{\alpha}\}$, where $\alpha$ ranges over all monomials in the indices $[n]$ of degree $t/2$.
\begin{enumerate}
	\item $v_i^2 = 1$ for all $i\in[n]$
	\item $-\vec{U}_{\alpha} \le (H^{\otimes t/2}v^{\otimes t/2})_{\alpha} \le \vec{U}_{\alpha}$ for all monomials $\alpha$ of degree $t/2$
	\item $\sum_{\alpha}\mu^{(\alpha)}\vec{U}_{\alpha} \le (\ell\log n + 1)^{t/2}$,
\end{enumerate}	where $\mu^{(\alpha)}\triangleq\prod_{i\in\alpha}\mu^{(i)}$.
\end{axioms}

As in the proof of Lemma~\ref{lem:satisfied}, the values of $\{\hat{\vp}_i\}$ and $\{w_i\}$ will be given by the canonical assignment, so the only variables in the SoS proofs of \eqref{eq:relaxedl1} and \eqref{eq:relaxedmoment} will be $v_1,...,v_n$ and, respectively, $\{\vec{W}_i\}_{i\in[n]}$ and $\{\vec{U}_{\alpha}\}_{|\alpha|\le t/2}$.

By definition, the existence of a degree-$d$ SoS proof for \eqref{eq:relaxedl1} using Axioms~\ref{axioms:l1_shape} is equivalent to the existence of polynomials $f^{K_1,K_2}_J(v,\vec{W},\{\hat{\vp}_i\},\hat{\vp})$ and $g^{K_1,K_2}_J(v,\vec{W},\{\hat{\vp}_i\},\hat{\vp})$ for $J,K_1,K_2\subset[n]$ for which \begin{multline}
	5\delta - \langle \hat{\vp}_i - \hat{\vp}, v\rangle = \\ 
	\sum_{J,K_1,K_2}f^{K_1,K_2}_J h^{K_1,K_2}_J + \left(\ell\log n + 1 - \sum_{i\in[n]}\mu^{(i)}\vec{W}_i\right)\sum_{J,K_1,K_2} g^{K_1,K_2}_J\cdot h^{K_1,K_2}_J,
\end{multline} where \begin{equation}
	h^{K_1,K_2}_J \triangleq \prod_{i\in J}(1 - v_i^2)\cdot \prod_{k_1\in K_1}(\vec{W}_{k_1} - (Hv)_{k_1}) \cdot \prod_{k_2\in K_2}(\vec{W}_{k_2} + (H_v)_{k_2}),
\end{equation} and where each $f^{K_1,K_2}_J$ and $g^{T_1,T_2}_J$ is a sum-of-squares polynomial such that $f^{K_1,K_2}_J\cdot h^{K_1,K_2}_J$ and $(\vec{W}_j - (Hv)_j) \cdot g^{K_1,K_2}_J \cdot h^{K_1,K_2}_J$ is degree $d$. We will take this degree to be $d = O(1)$.

Completely analogously, the existence of a degree-$d$ SoS proof for \eqref{eq:relaxedmoment} using Axioms~\ref{axioms:momentboundprime} is equivalent to the existence of polynomials $p^{T_1,T_2}_S(v,\vec{U},\{w_i\},\{\hat{\vp}_i\},\hat{\vp})$ and $q^{T_1,T_2}_S(v,U,\{w_i\},\{\hat{\vp}_i\},\hat{\vp})$ for $S\subset[n]$, $T_1,T_2\subseteq\{\alpha: |\alpha|\le t/2\}$ for which \begin{multline}
	2\cdot (8t/k)^{t/2}\sum_{i\in[N]}w_i - \sum_{i\in[N]}w_i\langle X_i - \hat{\vp}_i,v\rangle^t = \\ 
	\sum_{S,T_1,T_2}p^{T_1,T_2}_S r^{T_1,T_2}_S + \left((\ell\log n + 1)^{t/2} - \sum_{\alpha}\mu^{(\alpha)}\cdot\vec{U}_{\alpha}\right)\sum_{S,T_1,T_2} \cdot q^{T_1,T_2}_S\cdot r^{T_1,T_2}_S
\end{multline} where \begin{equation}
	r^{T_1,T_2}_S \triangleq \prod_{i\in S}\left(1 - v_i^2\right)\cdot \prod_{\alpha\in T_1}\left(\vec{U}_{\alpha} - (H^{\otimes t/2}v^{\otimes t/2})_{\alpha}\right)\cdot \prod_{\beta\in T_2}\left(\vec{U}_{\beta} + (H^{\otimes t/2}v^{\otimes t/2})_{\beta}\right)
\end{equation} and where each $p^{T_1,T_2}_S$ and $q^{T_1,T_2}_S$ is a sum-of-squares polynomial such that $p^{T_1,T_2}_S\cdot r^{T_1,T_2}_S$ and $\left(\vec{U}_{\alpha} - (H^{\otimes t/2}v^{\otimes t/2})_{\alpha}\right)\cdot q^{T_1,T_2}_S\cdot r^{T_1,T_2}_S$ is degree $d$. We will take this degree to be $d = O(t)$.

Let $F^{K_1,K_2}_J$ and $G^{K_1,K_2}_J$ respectively denote the matrices of coefficients of $f^{K_1,K_2}_J$ and $g^{K_1,K_2}_J$ as degree-$O(1)$ polynomials solely in the variables $\{v_i\}$ and $\{\vec{W}_i\}$, with entries denoted by $(F^{K_1,K_2}_J)_{\gamma,\rho}$ and $(G^{K_1,K_2}_J)_{\gamma,\rho}$. Likewise, let $P^{T_1,T_2}_S$ and $Q^{T_1,T_2}_S$ respectively denote the matrices of coefficients of $p^{T_1,T_2}_S$ and $q^{T_1,T_2}_S$ as degree-$O(t)$ polynomials solely in the variables $\{v_i\}$ and $\{\vec{U}_{\alpha}\}$, with entries denoted by $(P^{T_1,T_2}_S)_{\gamma,\rho}$ and $(Q^{T_1,T_2}_S)_{\gamma,\rho}$. 

\begin{remark}
	As we will demonstrate in the course of our analysis, we only need consider $K_1,K_2$ of size at most 1, and $T_1,T_2$ of size at most 2, so the total number of constraints in the overall program will only be singly-exponential in $t$.
	\label{remark:degree}
\end{remark}

We will consider the following program.

\begin{program}{$\hat{\calP}'$} The variables are $\{w_i\}_{i\in[N]}$, $\hat{\vp}$, $\{\hat{\vp}_i\}_{i\in[N]}$, $\{Q_{ij}\}$, $\{(P^{T_1,T_2}_S)_{\gamma,\rho}\}$, $\{(Q^{T_1,T_2}_S)_{\gamma,\rho}\}$, and the constraints are \begin{enumerate}
	\item $w^2_i = w_i$ for all $i\in[N]$\label{constraint:boolean_shape}
	\item $(1 - \epsilon)N\le \sum w_i \le (1 - \epsilon)N$ \label{constraint:cardinality_shape}
	\item \begin{multline}5\delta - \langle\hat{\vp}_i - \hat{\vp},v\rangle = \sum_{J,K_1,K_2}h^{K_1,K_2}_J \cdot \langle (v,\vec{W})^{\otimes t/2}, F^{J_1,J_2}_K(v,\vec{W})^{\otimes t/2}\rangle \\
	+ \left(\ell\log n + 1 - \sum_i \mu^{(i)}\vec{W}_i\right) \sum_{J,K_1,K_2}h^{K_1,K_2}_J \cdot \langle (v,\vec{W})^{\otimes t/2}, G^{J_1,J_2}_K(v,\vec{W})^{\otimes t/2}\rangle \end{multline}\label{constraint:l1sos_shape}
	\item $\sum_{i\in[N]}w_i X_i = \hat{\vp}\cdot\sum_{i\in[N]}w_i$ \label{constraint:mean_shape}
	\item \begin{multline}
	2\cdot (8t/k)^{t/2}\sum_{i\in[N]}w_i - \sum_{i\in[N]}w_i\langle X_i - \hat{\vp}_i,v\rangle^t = 
	\sum_{S,T_1,T_2}r^{T_1,T_2}_S\cdot\langle (v,\vec{U})^{\otimes t/2},P^{T_1,T_2}_S(v,\vec{U})^{\otimes t/2}\rangle \\ + \left((\ell\log n + 1)^{t/2} - \sum_{\alpha}\mu^{(\alpha)}\cdot \vec{U}_{\alpha}\right)\sum_{S,T_1,T_2}r^{T_1,T_2}_S\cdot \langle (v,\vec{U})^{\otimes t/2},Q^{T_1,T_2}_S(v,\vec{U})^{\otimes t/2}\rangle
\end{multline}\label{constraint:momentboundsos_shape}
	\item $F^{T_1,T_2}_S, G^{T_1,T_2}_S \succeq 0$ for all $T_1,T_2,S\subset[n]$ for which $|T_1|,|T_2|, |S|\le O(t)$.. \label{constraint:Qpsd_shape}
	\item $P^{T_1,T_2}_S, Q^{T_1,T_2}_S \succeq 0$ for all $T_1,T_2,S\subset[n]$ for which $|T_1|,|T_2|, |S|\le O(t)$. \label{constraint:MNpsd_shape}
	\item $\hat{\vp}_i \ge 0$ for all $i\in[n]$ and $\sum_i \hat{\vp}_i = 1$. \label{constraint:simplex_matrixsos_shape}
\end{enumerate}\label{program:matrixsos_shape}
\end{program}

\begin{proof}[Proof of Lemma~\ref{lem:satisfied_shape}]
	As before, solvability follows from the fact that the problem of outputting a degree-$O(t)$ pseudodistribution satisfying a system of degree-$O(t)$ polynomial constraints can be encoded a a semidefinite program of size $n^{O(t)}$.

	The fact that \ref{program:matrixsos_shape} $\sos{O(t)}$ \ref{program:shape_basic} follows by definition and by Lemma~\ref{eq:main_haar_inequality}.

	Finally, we verify that under the canonical assignment, with high probability over $X_1,...,X_N$ there exists a satisfying assignment to the remaining variables of \ref{program:matrixsos_shape}. As in the proof of Lemma~\ref{lem:satisfied}, the canonical assignment clearly satisfies Constraints \ref{constraint:boolean_shape}, \ref{constraint:cardinality_shape}, and \ref{constraint:mean_shape}.

	We prove that Constraints~\ref{constraint:l1sos_shape} and \ref{constraint:Qpsd_shape} are satisfiable with high probability in Lemma~\ref{lem:l1_shape_sat}, and we prove that Constraints~\ref{constraint:momentboundsos_shape} and \ref{constraint:MNpsd_shape} are satisfiable with high probability in Lemma~\ref{lem:moment_shape_sat}.
\end{proof}

The following fact will be useful in the proofs of Lemma~\ref{lem:l1_shape_sat} and \ref{lem:moment_shape_sat}.
	
\begin{lem}[``Shelling trick'']
	If $v\in\R^m$ satisfies $\norm{v}_2\le C$ and $\norm{v}_1 = C\cdot \sqrt{k}$, then there exist $k$-sparse vectors $v_1,...,v_{m/k}$ with disjoint supports for which $v = \sum^{m/k}_{i=1}v_i$ and $\sum^{m/k}_{i=1}\norm{v_i}_2 \le 2C$.\label{lem:shelling}
\end{lem}

\begin{proof}
	We may assume without loss of generality that $C = 1$. Let $B_1\subset[m]$ be the indices of the $k$ largest entries of $v$, $B_2$ be those of the next $k$ largest, and so on, so we may write $[m]$ as the disjoint union $B_1\cup\cdots\cup B_{m/k}$. For $i\in[m/k]$, define $v_i\in\R^m$ to be the restriction of $v$ to the coordinates indexed by $B_i$. For any $i$, note that for any $j\in B_i$, $|v_j| \le \frac{1}{k}\norm{v_{j-1}}_1$, so \begin{equation}
		\norm{v_i}^2_2 = \sum_{j\in B_i}v_j^2 \le k\cdot \frac{1}{k^2}\cdot\norm{v_{i-1}}_1^2 = \frac{1}{k}\norm{v_{i-1}}^2_1.
	\end{equation} So $\norm{v_i}_2 \le \norm{v_{i-1}}_1/\sqrt{k}$ and thus \begin{equation}
		\sum^{m/k}_{i=1}\norm{v_i}_2 \le \norm{v_1}_2 + \frac{1}{\sqrt{k}}\norm{v}_1 \le 2
	\end{equation} as desired.
\end{proof}

\begin{lem}
	Under the canonical assignment, with high probability there is some choice of $\{(F^{J_1,J_2}_K)_{\gamma,\rho}\}$ and $\{(G^{J_1,J_2}_K)_{\gamma,\rho}\}$ for which Constraints~\ref{constraint:l1sos_shape} and \ref{constraint:Qpsd_shape} are satisfied. \label{lem:l1_shape_sat}
\end{lem}

\begin{proof}
	We first write \begin{equation}\langle \hat{\vp}_i - \hat{\vp}, v\rangle = \frac{1}{m}\sum_{j\neq i}\langle \vp_i - \vp_j,v\rangle + \frac{1}{m}\sum_{j\in S_g}\langle X_j - \vp_j,v\rangle.\end{equation} Note that by Fact~\ref{fact:1bounded}, \begin{equation}\{v_i^2 = 1 \ \forall \ 1\le i\le n\} \sos{2} \langle p_i - p_j, v\rangle \le \norm{\vp_i - \vp_j}_1 \le \norm{\vp_i - \vp}_1 + \norm{\vp_j - \vp}_1 \le 4\delta.\end{equation} It remains to show that with high probability, there is a degree-$O(t)$ proof that Axioms~\ref{axioms:l1_shape} imply $\frac{1}{m}\sum_{j\in S_g}\langle X_j - \vp_j,v\rangle\le \delta$.

	Equivalently, we must show that for any degree-$t$ pseudodistribution $\psE$ over the variables $v$ and $\vec{U}$ which satisfies Axioms~\ref{axioms:l1_shape}, we have that \begin{equation}\frac{1}{m}\sum_{j\in S_g}\langle X_j - \vp_j, \psE[v]\rangle \le \delta.\label{eq:akempiricalbound}\end{equation} The set of vectors $\psE[v]$ arising from pseudodistributions $\psE$ satisfying Axioms~\ref{axioms:l1_shape} is some convex set $\calJ\subset\R^n$.

	\begin{lem}
		Let $\calJ$ be the convex set of all vectors of the form $\psE[v]$ for some degree-$t$ pseudodistribution $\psE$ over the variables $v,\vec{W}$ satisfying Axioms~\ref{axioms:l1_shape}.

		Additionally, let $\calJ_1,\calJ_2\subset\R^n$ consist of all vectors $u$ for which $\sum_i \mu^{(i)}|u_i| \le \ell\log n + 1$ and for which $\norm{u}_2 \le \sqrt{n}$ respectively. Then \begin{equation}
			\calJ\subset H^{-1}(\calJ_1 \cap \calJ_2)
		\end{equation}\label{lem:contain_l1}
	\end{lem}

	\begin{proof}
		Take any $u\in\calJ$. We first show that $u\in H^{-1}\cdot\calJ_1$. By linearity of $\psE$, we may write $u$ as \begin{equation}
			u = H^{-1}\cdot\psE[Hv].
		\end{equation} For any $i\in[n]$, the second of Axioms~\ref{axioms:l1_shape} immediately implies that \begin{equation}
			-\vec{W}_i \le \psE\left[(Hv)_i\right] \le \vec{W}_i.
		\end{equation} We emphasize that this is the only place where we use the second of Axioms~\ref{axioms:l1_shape}, and only in a linear fashion, hence Remark~\ref{remark:degree}.

		So $\sum_i\mu^{(i)}|(Hu)_i| \le \psE[\sum_i \mu^{(i)}\vec{W}_i]\le \ell\log n + 1$, where the last inequality follows by the third of Axioms~\ref{axioms:l1_shape}.

		Finally, to show that $u\in H^{-1}\cdot\calJ_2$, note first that by orthonormality of $H$, it is enough to show that $u\in\calJ_2$. But this follows immediately from the fact that $\psE$ satisfies the first of Axioms~\ref{axioms:l1_shape}, which by \eqref{eq:trivial} implies that $-1\le \psE[v_i]\le 1$ for all $i\in[n]$, from which we conclude that $\norm{u}^2_2= n$ and thus $u\in\calJ_2$.
	\end{proof}

	\begin{lem}
		For every $\eta \le (\ell\log n + 1)^{-1}$, there exists a set $\mathcal{N}\subset\P_{n-1}(\R)$ of size $O(n^{3/2}/\eta)^s$ such that for every $u\in H^{-1}(\calJ_1\cap \calJ_2)$, there exists some $\tilde{u} = \sum_{\nu} \alpha_{\nu}\cdot u^*_{\nu}$ for $u^*_{\nu} \in\mathcal{N}$ such that 1) $\norm{u - \tilde{u}}_2 \le \eta$, 2) $\sum_{\nu} \alpha_{\nu} \le 1$, and 3) $\norm{u^*_{\nu}}_{\infty} \le 2(\ell\log n + 1)$ for all $\nu$.\label{lem:sumnet_l1}
	\end{lem}

	\begin{proof}
		Let $s = \ell\log n + 1$, and let $m = \log n$. Let $\mathcal{N}'$ be an $\frac{\eta}{(m+1)\sqrt{n}}$-net in $L_2$ norm for all $s^2$-sparse vectors in $\S^{n - 1}$. Because $\S^{s^2 - 1}$ has an $\frac{\eta}{(m+1)\sqrt{n}}$-net in $L_2$ norm of size $(3(m+1)\sqrt{n}/\eta)^{s^2}$, by a union bound we have that $|\mathcal{N}'|\le \binom{n}{s^2}\cdot(3(m+1)\sqrt{n}/\eta)^{s^2} = O(n^{3/2}\log n/\eta)^{s^2}$.

		Take any $u\in H^{-1}(\calJ_1\cap \calJ_2)$ and consider $w \triangleq Hu\in\calJ_1\cap \calJ_2$. We may write $w$ as $\sum_{i\in\mathcal{T}}w[i]$, where \begin{equation}w[i] = \sum_{\nu\in T_i}w_{\nu}\cdot e_{\nu}\label{eq:levels}\end{equation} for $e_{\nu}$ the $\nu$-th standard basis vector in $\R^n$.

		As the nonzero entries of $w[i]$ are just a subset of those of $w$, we clearly have $\norm{w[i]}_2 \le \sqrt{n}$ for all $i\in\mathcal{T}$. Moreover, because $w\in\calJ_1$, we have that \begin{equation}\sum_i 2^{-(m-i)/2}|w[i]| \le s,\label{eq:weightedl1}\end{equation} so in particular \begin{equation}\norm{w[i]}_1 \le 2^{(m-i)/2}\cdot s = 2^{-i/2}\cdot s\sqrt{n}.\label{eq:wil1bound}\end{equation} We can thus apply Lemma~\ref{lem:shelling} to conclude that for each $i\in[m]$, $w[i] = \sum_j w^{i,j}$ for some vectors $\{w^{i,j}\}_j$ of sparsity at most $\lceil 2^{-i}\cdot s^2\rceil \le s^2$ and for which \begin{equation}
			\sum_j \norm{w^{i,j}}_2 \le \sqrt{n}. \label{eq:sumcoeffs}
		\end{equation} For each $w^{i,j}$, there is some $(w')^{i,j}\in\mathcal{N}'$ such that if we define $\tilde{w}^{i,j} \triangleq \norm{w^{i,j}}_2 \cdot (w')^{i,j}$, then we have \begin{equation}\norm{w^{i,j} - \tilde{w}^{i,j}}_2 \le \frac{\eta}{(m+1)\sqrt{n}}\cdot\norm{w^{i,j}}_2.\label{eq:net_diff}\end{equation} Defining $\tilde{w}[i]\triangleq \sum_j \tilde{w}^{i,j}$, we get that \begin{equation}
			\norm{w[i] - \tilde{w}[i]}_2 \le \frac{\eta}{(m+1)\sqrt{n}}\sum_j \norm{w^{i,j}}_2 \le \frac{\eta}{m+1}.
		\end{equation} So if we define $\tilde{w} \triangleq \sum_{i\in\mathcal{T}}\tilde{w}[i] = \sum_{i\in\mathcal{T}}\sum_j \tilde{w}^{i,j}$, we have that $\norm{w - \tilde{w}}_2 \le \eta$.

		Now let $\mathcal{N}\triangleq \P\left(H^{-1}\mathcal{N}'\right)$. As $u = H^{-1}w$ and $H^{-1}$ is an isometry, if we define $\tilde{u}^{i,j} \triangleq H^{-1}\tilde{w}^{i,j}$ and $\tilde{u} \triangleq \sum_{i\in\mathcal{T}}\sum_j\tilde{u}^{i,j}$, then we likewise get that $\norm{u - \tilde{u}}_2 \le \eta$, and clearly $\tilde{u}^{i,j}\in\mathcal{N}$, concluding the proof of part 1) of the lemma.

		For each $\tilde{u}^{i,j}$, define \begin{equation}u^{i,j}_*\triangleq \tilde{u}^{i,j} / \alpha_{i,j} \ \ \ \text{for} \ \ \ \alpha_{i,j}\triangleq s^{-1}\cdot 2^{-(m-i)/2}\norm{w^{i,j}}_{\infty}\label{eq:ustardef}\end{equation} so that \begin{equation}\tilde{u}= \sum_{i,j}\alpha_{i,j}u^{i,j}_*.\label{eq:udecomp}\end{equation} Note that \begin{align*}
			\sum_{i,j}\alpha_{i,j}&\le \frac{1}{s}\sum_i 2^{-(m-i)/2}\sum_j\norm{w^{i,j}}_{\infty} \\
			&\le \frac{1}{s}\sum_i 2^{-(m-i)/2}\norm{w[i]}_1 \\
			&\le \frac{1}{s}\cdot s = 1,
		\end{align*} where the second inequality follows by the fact that for fixed $i$, the supports of the vectors $w^{i,j}$ are disjoint for different $j$ so that $\sum_j \norm{w^{i,j}}_{\infty}\le \norm{w[i]}_1$, and the third inequality follows from \eqref{eq:weightedl1}. This concludes the proof of part 2) of the lemma.

		Finally, we need to bound $\norm{u^{i,j}_*}_{\infty}$. Note first that for any vector $z$ supported only on indices $\nu\in T_i$, \begin{equation}\norm{H^{-1}z}_{\infty} \le 2^{-(m-i)/2}\cdot\norm{z}_{\infty}\label{eq:Hinvwij}\end{equation} because the Haar wavelets $\{\psi_{i,j}\}_j$ have disjoint supports and $L_{\infty}$ norm $2^{-(m-i)/2}$. It follows that \begin{align*}
			\norm{\tilde{u}^{i,j}}_{\infty} &\le \norm{H^{-1}w^{i,j}}_{\infty} + \norm{H^{-1}(w^{i,j} - \tilde{w}^{i,j})}_{\infty} \\
			&\le 2^{-(m-i)/2}\cdot\norm{w^{i,j}}_{\infty} + 2^{-(m-i)/2}\norm{w^{i,j} - \tilde{w}^{i,j}}_{\infty} \\
			&\le 2^{-(m-i)/2}\cdot\norm{w^{i,j}}_{\infty} + 2^{-(m-i)/2}\norm{w^{i,j} - \tilde{w}^{i,j}}_{2} \\
			&\le 2^{-(m-i)/2}\cdot\norm{w^{i,j}}_{\infty} + 2^{-(m-i)/2}\cdot\frac{\eta}{(m+1)\sqrt{n}}\norm{w^{i,j}}_2 \\
			&\le 2^{-(m-i)/2}\cdot\norm{w^{i,j}}_{\infty} + 2^{-(m-i)/2}\cdot\frac{\eta}{(m+1)\sqrt{n}}\norm{w^{i,j}}_{\infty}\cdot s \\
			&= 2^{-(m-i)/2}\cdot\norm{w^{i,j}}_{\infty}\left(1 + \frac{\eta\cdot s}{(m+1)\sqrt{n}}\right) \\
			&\le 2\cdot 2^{-(m-i)/2}\cdot\norm{w^{i,j}}_{\infty},
		\end{align*} where the first inequality is triangle inequality, the second inequality follows by \eqref{eq:Hinvwij}, the third inequality follows from monotonicity of $L_p$ norms, the fourth inequality follows from \eqref{eq:net_diff},  the fifth inequality follows from the fact that $w^{i,j}$ is $s^2$-sparse, and the final inequality follows from the hypothesis that $\eta\le 1/s$. Recalling \eqref{eq:ustardef}, we conclude that $\norm{u^{i,j}_*}_{\infty}\le 2s$ as claimed.
	\end{proof}

	Next we show that we can control $\frac{1}{m}\sum_{j\in S_g}\langle X_j - \vp_j, u\rangle$ for all directions $u$ in the net $\mathcal{N}$.

	\begin{lem}
		Let $\xi > 0$ and let $\mathcal{N}\in\P_{n-1}(\R)$ be any collection of $M$ directions. Then \begin{equation}
			\Pr\left[\frac{1}{m}\sum_{j\in S_g}\langle X_j - \vp_j, u\rangle > \xi\cdot\norm{u}_{\infty} \ \forall u\in \mathcal{N}\right] < 2M\cdot e^{-2m\xi^2}, \label{eq:concnet} 
		\end{equation} where the probability is over the samples $X_j$ for $j\in S_g$.\label{lem:shell_conc}
	\end{lem}

	\begin{proof}
		Without loss of generality, assume that $\norm{u}_{\infty} = 1$. For any $j\in S_g$, note that $\norm{X_j - \vp_j}_1 \le 2$, so $\langle X_j - \vp_j, u\rangle$ is a $[-2,2]$-valued random variable, call it $A_j$. By Hoeffding's inequality, \begin{equation}
			\Pr\left[\left|\frac{1}{m}\sum_{j\in S_g}A_j - \frac{1}{m}\sum_{j\in S_g}\E[A_j]\right|\ge \xi\right] \le 2e^{-2m\xi^2},
		\end{equation} so we are done by a union bound over the $M$ directions in $\mathcal{N}$.
	\end{proof}

	We may now proceed with the proof of \eqref{eq:akempiricalbound}. For $u\in\calJ$, by Lemmas~\ref{lem:contain_l1} and \ref{lem:sumnet_l1}, there is some $\tilde{u} = \sum_{\nu}\alpha_{\nu}u^*_{\nu}$ such that $u^*_{\nu}\in\mathcal{N}$ and $\norm{u - \tilde{u}}_2 \le \eta$. We may write \begin{align*}
		\frac{1}{m}\sum_{j\in S_g}\langle X_j - \vp_j, u\rangle &\le \frac{1}{m}\sum_{j\in S_g}\langle X_j - \vp_j, \tilde{u}\rangle + \norm{\frac{1}{m}\sum_{j\in S_g}X_j}_2\cdot\norm{u - \tilde{u}}_2 \\
		&\le \frac{1}{m}\sum_{j\in S_g}\langle X_j - \vp_j, \tilde{u}\rangle + \eta \\
		&= \sum_{\nu}\alpha_{\nu}\left(\frac{1}{m}\sum_{j\in S_g}\langle X_j - \vp_j, u^*_{\nu}\rangle\right) + \eta \\
		&\le \sum_{\nu}\alpha_{\nu}\cdot \xi\cdot\norm{u^*_{\nu}}_{\infty} + \eta \\
		&\le 2\xi(\ell\log n + 1)(\log n + 1) + \eta,
	\end{align*} where the second inequality follows from the fact that $\frac{1}{m}\sum_{j\in S_g}X_j$ is a vector in $\Delta^n$ and thus has $L_2$ norm at most 1, and the penultimate step holds with probability $2|\mathcal{N}|e^{-8m\xi^2}$.

	So if $\eta = \delta/2$ and $\xi = \frac{\delta}{4(\ell\log n + 1)(\log n + 1)}$, then as long as \begin{equation}m = \Omega(\xi^{-2}\log|\mathcal{N}|) = \Omega\left(\frac{\log(1/\delta)}{\delta}\cdot\ell^4\log^7 n\right),\end{equation} then with probability at least $1 - \poly(n)$, there exists an SoS proof of \eqref{eq:akempiricalbound} using Axioms~\ref{axioms:l1_shape}.
\end{proof}

\begin{lem}
	Under the canonical assignment, with high probability there is some choice of $\{(P^{T_1,T_2}_S)_{\gamma,\rho}\}$ and $\{(Q^{T_1,T_2}_S)_{\gamma,\rho}\}$ for which Constraints~\ref{constraint:momentboundsos_shape} and \ref{constraint:MNpsd_shape} are satisfied.
	\label{lem:moment_shape_sat}
\end{lem}

The proof of Lemma~\ref{lem:moment_shape_sat} is conceptually very similar to that of Lemma~\ref{lem:l1_shape_sat}, so we defer it to Appendix~\ref{app:defer}.

	\bibliographystyle{alpha}
	\bibliography{biblio}

	\appendix


\section{Proof of Lemma~\ref{lem:moment_shape_sat}}
\label{app:defer}

In the arguments that follow, it will be useful to define the notion of projectivization. Given a set $S\subset\R^m$, let $\P S$ denote its projectivization, namely the quotient of $S$ by the equivalence relation $u\sim v$ if $u = \lambda v$ for some $\lambda\in\R$. We will denote the projectivization of $\R^m$ by $\P_{n-1}(\R)$. Occasionally we will abuse notation and implicitly associate $S\subset\P_{n-1}(\R)$ with its fiber under the quotient map $\R^n\to\P_{n-1}(\R)$.

\begin{proof}[Proof of Lemma~\ref{lem:moment_shape_sat}]
As in the proof of Lemma~\ref{lem:satisfied}, because of Lemma~\ref{lem:basicbound} it is enough to show an SoS proof using Axioms~\ref{axioms:momentboundprime} that \begin{equation}
	\frac{1}{m}\sum_{i\in S_g}\langle X_i - \vp_i,v\rangle^t - \frac{1}{m}\sum_{i\in S_g}\E_{X\sim \mathcal{D}_i}\langle X - \vp_i,v\rangle^t \le (8t/k)^{t/2}.\label{eq:relaxedmoment2}
\end{equation} Equivalently, we must show that for any degree-$t$ pseudodistribution $\psE$ over the variables $v$ and $\vec{U}$ which satisfies Axioms~\ref{axioms:momentboundprime}, we have that \begin{equation}
	\left\langle \vec{Z}, \psE\left[v^{\otimes t/2}(v^{\otimes t/2})^{\top}\right]\right\rangle\le (8t/k)^{t/2},\label{eq:relaxedmoment3}
\end{equation} where $\vec{Z}\triangleq \vec{Z}[S_g]$. The set of matrices $\psE[v^{\otimes t/2}(v^{\otimes t/2})^{\top}]$ arising from pseudodistributions $\psE$ satisfying Axioms~\ref{axioms:momentboundprime} is some convex set $\mathcal{K}$ in $\R^{n^{t/2}\times n^{t/2}}$.


\begin{lem}
	Let $\mathcal{K}\subset\R^{n^{t/2}\times n^{t/2}}$ be the convex set of all matrices of the form $\psE[v^{\otimes t/2}(v^{\otimes t/2})^{\top}]$ for some degree-$t$ pseudodistribution $\psE$ over the variables $v,\vec{U}$ satisfying Axioms~\ref{axioms:momentboundprime}.

	Additionally, let $\mathcal{K}_1,\mathcal{K}_2\subset\R^{n^{t/2}\times n^{t/2}}$ consist of all matrices $\vec{M}$ for which $\sum_{\alpha,\beta} \mu^{(\alpha)}\mu^{(\beta)}|\vec{M}_{\alpha,\beta}| \le (\ell\log n + 1)^t$ and for which $\norm{\vec{M}}_F\le n^{t/2}$ respectively. Then \begin{equation}\mathcal{K}\subset \left[(H^{-1})^{\otimes t/2}\right](\mathcal{K}_1\cap \mathcal{K}_2)\left[(H^{-1})^{\otimes t/2}\right]^{\top},\end{equation}\label{lem:contain}
\end{lem}

\begin{proof}
	Take any $\vec{M}\in\mathcal{K}$. We first show that $\vec{M}\in\left[(H^{-1})^{\otimes t/2}\right]\mathcal{K}_1\left[(H^{-1})^{\otimes t/2}\right]^{\top}$. By linearity of $\psE$, we may write $\vec{M}$ as \begin{equation}
		\vec{M} = (H^{-1})^{\otimes t/2}\cdot \psE\left[\left[H^{\otimes t/2}v^{\otimes t/2}\right]\cdot \left[H^{\otimes t/2}v^{\otimes t/2}\right]^{\top}\right]\cdot ((H^{-1})^{\otimes t/2})^{\top}.
	\end{equation} For any monomials $\alpha,\beta$ each of degree $t/2$, the second of Axioms~\ref{axioms:momentboundprime} immediately implies that \begin{equation}-\vec{U}_{\alpha}\vec{U}_{\beta} \le \psE\left[\left[H^{\otimes t/2}v^{\otimes t/2}\right]_{\alpha}\cdot \left[H^{\otimes t/2}v^{\otimes t/2}\right]_{\beta}\right] \le \vec{U}_{\alpha}\vec{U}_{\beta}.\end{equation} We emphasize that this is the only place where we use the second of Axioms~\ref{axioms:momentboundprime}, and only in a degree-2 fashion, hence Remark~\ref{remark:degree}. So $\sum_{\alpha,\beta}\mu^{(\alpha)}\mu^{(\beta)}|\vec{M}_{\alpha,\beta}| \le \psE\left[\sum_{\alpha,\beta}\mu^{(\alpha)}\mu^{(\beta)}\vec{U}_{\alpha}\vec{U}_{\beta}\right] \le (\ell log n + 1)^t$, where the last inequality follows by axiom 3.

	Finally, to show that $\vec{M}\in \left[(H^{-1})^{\otimes t/2}\right]\mathcal{K}_2\left[(H^{-1})^{\otimes t/2}\right]^{\top}$, note first that by orthonormality of $H$, it is enough to show that $\vec{M}\subset\mathcal{K}_2$. But this follows immediately from the fact that $\psE$ satisfies the first of Axioms~\ref{axioms:momentboundprime}. Indeed, from Fact~\ref{fact:1bounded} and the fact that $\psE$ is degree-$O(t)$ we get that $-1\le \psE[v_{\alpha}v_{\beta}]\le 1$, so \begin{equation}
		\sum_{|\alpha|,|\beta| = t/2}\vec{M}^2_{\alpha,\beta} = \sum_{\alpha,\beta}\psE[v_{\alpha}v_{\beta}]^2 \le n^t
	\end{equation} as claimed.
\end{proof}

\begin{lem}
	For every $\eta\le (\ell\log n + 1)^{-1}$, there exists a set $\mathcal{N}\subset \P(\R^{n^{t/2}\times n^{t/2}})$ such that for every $\vec{M}\in\left[(H^{-1})^{\otimes t/2}\right](\mathcal{K}_1\cap \mathcal{K}_2)\left[(H^{-1})^{\otimes t/2}\right]^{\top}$, there exists some $\tilde{\vec{M}} = \sum_{\nu} \alpha_{\nu}\cdot\vec{M}^*_{\nu}$ for $\vec{M}^*_{\nu}\in\mathcal{N}$ such that 1) $\norm{\vec{M} - \tilde{\vec{M}}}_F \le \eta$, 2) $\sum_{\nu}\alpha_{\nu}\le 1$, and 3) $\norm{\vec{M}^*_{\nu}}_{\max}\le 2(\ell\log n + 1)^t$.\label{lem:sumnet}
\end{lem}

\begin{proof}
	Let $s = \ell\log n + 1$, and let $m = \log n$. Let $\mathcal{N}'$ be an $\frac{\eta}{(m+1)^tn^{t/2}}$-net in Frobenius norm for all $s^{2t}$-sparse $n^{t/2}\times n^{t/2}$ matrices of unit Frobenius norm. Because $\S^{s^{2t} - 1}$ has an $\frac{\eta}{m\sqrt{n}}$-net in $L_2$ norm of size $(3(m+1)^tn^{t/2}/\eta)^{s^{2t}}$, by a union bound we have that \begin{equation}
		|\mathcal{N}'| \le \binom{n^t}{s^{2t}}\cdot (3(m+1)^t n^{t/2}/\eta)^{s^{2t}} = O(n^{3t/2}\log^t n/\eta)^{s^{2t}}
	\end{equation}

	Take any $\vec{M}\in \left[(H^{-1})^{\otimes t/2}\right](\mathcal{K}_1\cap \mathcal{K}_2)\left[(H^{-1})^{\otimes t/2}\right]^{\top}$ and consider $\vec{L}\triangleq H^{\otimes t/2}\vec{M}\left[H^{\otimes t/2}\right]^{\top}$. Define $\mathcal{T} \triangleq \{0_{\father}, 0_{\mother}, 1,...,m-1\}$. We may write $\vec{L}$ as $\sum_{\sigma,\tau}\vec{L}[\sigma,\tau]$, where $\sigma,\tau$ are monomials of degree $t/2$ in the indices $\mathcal{T}$, and where $\vec{L}[\sigma,\tau]$ the submatrix of $\vec{L}$ consisting of all entries from the rows $\alpha$ (resp. columns $\beta$) for which $\alpha_i\in T_{\sigma_i}$ (resp. $\beta_i \in T_{\tau_i}$) for all $1\le i\le t/2$.

	As the nonzero entries of $\vec{L}[\sigma,\tau]$ are just a subset of those of $\vec{L}$, we clearly have $\norm{\vec{L}[\sigma,\tau]}_F\le n^{t/2}$ for all $\sigma,\tau$. Moreover, because $\vec{L}\in\calK_1$, we have that \begin{equation}\sum_{\sigma,\tau}\prod^{t/2}_{i=1}2^{-(m-\sigma_i)/2}\cdot \prod^{t/2}_{j=1}2^{-(m-\tau_j)/2}\cdot \norm{\vec{L}[\sigma,\tau]}_{1,1} \le s^t\label{eq:weighted_L1_moment}\end{equation} so in particular \begin{equation}
		\norm{\vec{L}[\sigma,\tau]}_{1,1} \le \prod^{t/2}_{i=1} 2^{(m-\sigma_i)/2} \cdot \prod^{t/2}_{j=1}2^{(m - \tau_j)/2}\cdot s^t = 2^{-(\sum_i \sigma_i + \sum_j \tau_j)/2}\cdot s^t\cdot n^{t/2}.\label{eq:Lstl1bound}
	\end{equation} We can thus apply Lemma~\ref{lem:shelling} to conclude that for each $\sigma,\tau$, $\vec{L}[\sigma,\tau] = \sum_j \vec{L}^{\sigma,\tau;j}$ for some matrices $\{\vec{L}^{\sigma,\tau;j}\}_j$ of sparsity at most $\ceil 2^{-\sum_i \sigma_i - \sum_j \tau_j}\cdot s^{2t} \le s^{2t}$ and for which \begin{equation}
		\sum_j \norm{\vec{L}^{\sigma,\tau;j}}_F \le n^{t/2}.
	\end{equation} For each $\vec{L}^{\sigma,\tau;j}$, there is some $(\vec{L}')^{\sigma,\tau;j}\in\mathcal{N}'$ such that if we define $\tilde{\vec{L}}^{\sigma,\tau;j}\triangleq \norm{\vec{L}^{\sigma,\tau;j}}_F \cdot (\vec{L}')^{\sigma,\tau;j}$, then we have \begin{equation}
		\norm{\vec{L}^{\sigma,\tau;j} - \tilde{\vec{L}}^{\sigma,\tau;j}}_F \le \frac{\eta}{(m+1)^tn^{t/2}}\cdot\norm{\vec{L}^{\sigma,\tau;j}}_F.\label{eq:net_diff_moment}
	\end{equation} Defining $\tilde{\vec{L}}[\sigma,\tau]\triangleq \sum_j \tilde{\vec{L}}^{\sigma,\tau;j}$, we get that \begin{equation}
		\norm{\vec{L}[\sigma,\tau] - \tilde{\vec{L}}[\sigma,\tau]}_F \le \frac{\eta}{(m+1)^tn^{t/2}}\sum_j \norm{\vec{L}^{\sigma,\tau;j}}_F \le \frac{\eta}{(m+1)^t}.
	\end{equation} So if we define $\tilde{\vec{L}} = \sum_{\sigma,\tau}\tilde{\vec{L}}[\sigma,\tau] = \sum_{\sigma,\tau}\sum_j \tilde{\vec{L}}^{\sigma,\tau;j}$, we have that $\norm{\vec{L} - \tilde{\vec{L}}}_F \le \eta$.

	Now let $\mathcal{N}\triangleq \left((H^{-1})^{\otimes t/2} \mathcal{N}' \left[(H^{-1})^{\otimes t/2}\right]\right)$. As $\vec{M} = (H^{-1})^{\otimes t/2} \vec{L}\left[(H^{-1})^{\otimes t/2}\right]$ and $(H^{-1})^{\otimes t/2}$ is an isometry, if we define $\tilde{\vec{M}}^{\sigma,\tau;j}\triangleq (H^{-1})^{\otimes t/2} \tilde{\vec{L}}^{\sigma,\tau;j}\left[(H^{-1})^{\otimes t/2}\right]$ and $\tilde{\vec{M}} \triangleq \sum_{\sigma,\tau}\sum_j \tilde{\vec{M}}^{\sigma,\tau;j}$, then we likewise get that $\norm{\vec{M} - \tilde{\vec{M}}}_F \le \eta$, and clearly $\tilde{\vec{M}}^{\sigma,\tau;j}\in\mathcal{N}$, concluding the proof of part 1) of the lemma.

	For each $\tilde{\vec{M}}^{\sigma,\tau;j}$, define \begin{equation}
		\vec{M}^{\sigma,\tau;j}_* \triangleq \tilde{\vec{M}}^{\sigma,\tau;j}/\alpha_{\sigma,\tau;j} \ \ \ \text{for} \ \ \ \alpha_{\sigma,\tau;j} \triangleq s^{-t}\cdot \prod^{t/2}_{i=1}2^{-(m-\sigma_i)/2}\cdot \prod^{t/2}_{j=1}2^{-(m-\sigma_j)/2}\norm{\vec{L}^{\sigma,\tau;j}}_{\max}\label{eq:Mstardef}
	\end{equation} so that \begin{equation}
		\vec{M} = \sum_{\sigma,\tau,j}\alpha_{\sigma,\tau;j}\vec{M}^{\sigma,\tau;j}_*. \label{eq:Mdecomp}
	\end{equation} Note that \begin{align*}
		\sum_{\sigma,\tau,j}\alpha_{\sigma,\tau;j} &\le \frac{1}{s^t}\sum_{\sigma,\tau}\prod^{t/2}_{i=1}2^{-(m-\sigma_i)/2}\cdot \prod^{t/2}_{j=1}2^{-(m-\sigma_j)/2}\sum_j \norm{\vec{L}^{\sigma,\tau;j}}_{\max} \\
		&\le \frac{1}{s^t}\sum_{\sigma,\tau}\prod^{t/2}_{i=1}2^{-(m-\sigma_i)/2}\cdot \prod^{t/2}_{j=1}2^{-(m-\sigma_j)/2}\norm{\vec{L}[\sigma,\tau]}_{1,1} \\
		&\le \frac{1}{s^t}\cdot s^t = 1,
	\end{align*} where the second inequality follows by the fact that for fixed $\sigma,\tau$, the supports of the matrices $\vec{L}^{\sigma,\tau;j}$ are disjoint for different $j$ so that $\sum_j\norm{\vec{L}^{\sigma,\tau;j}}_{\max} \le \norm{\vec{L}[\sigma,\tau]}_{1,1}$, and the third inequality follows from \eqref{eq:Lstl1bound}. This concludes the proof of part 2) of the lemma.

	Finally, we need to bound $\norm{\vec{M}^{\sigma,\tau;j}_*}_{\max}$. Note first that for any matrix $\vec{J}$ supported only on the support of some $\vec{L}[\sigma,\tau]$, \begin{equation}
		\norm{(H^{-1})^{\otimes t/2} \vec{J}\left[(H^{-1})^{\otimes t/2}\right]}_{\max} \le \prod^{t/2}_{i=1}2^{-(m-\sigma_i)/2}\cdot\prod^{t/2}_{j=1}2^{-(m-\tau_j)/2}\cdot \norm{\vec{J}}_{\max}\label{eq:Hinvtensor}
	\end{equation} because the tensored Haar wavelets $\{\psi_{\sigma_1,j_1}\otimes\cdot\otimes \psi_{\sigma_{t/2},j_{t/2}}\}_{j_1,...,j_{t/2}}$ (resp. $\{\psi_{\tau_1,j_1}\otimes\cdot\otimes \psi_{\tau_{t/2},j_{t/2}}\}_{j_1,...,j_{t/2}}$) have disjoint supports and max-norm $\prod^{t/2}_{i=1}2^{-(m-\sigma_i)/2}$ (resp. $\prod^{t/2}_{j=1}2^{-(m-\tau_j)/2}$). It follows that \begin{align*}
		\norm{\tilde{\vec{M}}^{\sigma,\tau;j}}_{\max} &\le \norm{(H^{-1})^{\otimes t/2}\vec{L}^{\sigma,\tau;j}\left[(H^{-1})^{\otimes t/2}\right]^{\top}}_{\max} + \norm{(H^{-1})^{\otimes t/2}\left(\vec{L}^{\sigma,\tau;j} - \tilde{\vec{L}}^{\sigma,\tau;j}\right)\left[(H^{-1})^{\otimes t/2}\right]^{\top}}_{\max} \\
		&\le \prod^{t/2}_{i=1}2^{-(m-\sigma_i)/2}\prod^{t/2}_{j=1}2^{-(m-\tau_j)/2}\cdot\left(\norm{\vec{L}^{\sigma,\tau;j}}_{\max} + \norm{\vec{L}^{\sigma,\tau;j} - \tilde{\vec{L}}^{\sigma,\tau;j}}_{\max}\right) \\
		&\le \prod^{t/2}_{i=1}2^{-(m-\sigma_i)/2}\prod^{t/2}_{j=1}2^{-(m-\tau_j)/2}\cdot\left(\norm{\vec{L}^{\sigma,\tau;j}}_{\max} + \norm{\vec{L}^{\sigma,\tau;j} - \tilde{\vec{L}}^{\sigma,\tau;j}}_{F}\right) \\ 
		&\le \prod^{t/2}_{i=1}2^{-(m-\sigma_i)/2}\prod^{t/2}_{j=1}2^{-(m-\tau_j)/2}\cdot\left(\norm{\vec{L}^{\sigma,\tau;j}}_{\max} + \frac{\eta}{(m+1)^tn^{t/2}}\norm{\vec{L}^{\sigma,\tau;j}}_F\right) \\
		&\le \prod^{t/2}_{i=1}2^{-(m-\sigma_i)/2}\prod^{t/2}_{j=1}2^{-(m-\tau_j)/2}\cdot\left(\norm{\vec{L}^{\sigma,\tau;j}}_{\max} + \frac{\eta}{(m+1)^tn^{t/2}}\norm{\vec{L}^{\sigma,\tau;j}}_{\max}\cdot s^t\right) \\
		&= \prod^{t/2}_{i=1}2^{-(m-\sigma_i)/2}\prod^{t/2}_{j=1}2^{-(m-\tau_j)/2}\cdot\norm{\vec{L}^{\sigma,\tau;j}}_{\max}\cdot\left(1 + \frac{\eta\cdot s^t}{(m+1)^tn^{t/2}}\right) \\
		&\le 2\cdot \prod^{t/2}_{i=1}2^{-(m-\sigma_i)/2}\prod^{t/2}_{j=1}2^{-(m-\tau_j)/2}\cdot \norm{\vec{L}^{\sigma,\tau;j}}_{\max},
	\end{align*} where the first inequality is triangle inequality, the second inequality follows by \eqref{eq:Hinvtensor}, the third inequality follows from monotonicity of $L_p$ norms, the fourth inequality follows from \eqref{eq:net_diff_moment}, the fifth inequality follows from the fact that $\vec{L}^{\sigma,\tau;j}$ is $s^{2t}$ sparse, and the final inequality follows from the hypothesis that $\eta \le s^{-t}$. Recalling \eqref{eq:Mstardef}, we conclude that $\norm{\vec{M}^{\sigma,\tau;j}_*}_{\max} \le 2s^t$ as claimed.
\end{proof}

Next we show that we can control $\langle\vec{Z},\vec{M}\rangle$ for all directions in the net $\mathcal{N}$.

\begin{lem}
	Let $\xi > 0$ and let $\mathcal{N}\in\P\left(\R^{n^{t/2}\times n^{t/2}}\right)$ be any collection of $M$ directions. Then \begin{equation}
		\Pr\left[\langle\vec{Z},\vec{M}\rangle > \xi\cdot\norm{\vec{M}}_{\max} \ \forall \ \vec{M}\in\mathcal{N}\right] < 2M\cdot e^{-8m\xi^2},
	\end{equation} where the probability is over the samples $X_j$ for $j\in S_g$.\label{lem:shell_conc_shape}
\end{lem}

\begin{proof}
	Without loss of generality, assume that $\norm{\vec{M}}_{\max} = 1$. For any $j\in S_g$, note that the sum of the absolute values of the entries of the matrix $\left[(X - \vp_i)^{\otimes t/2}\right]\left[(X - \vp_i)^{\otimes t/2}\right]^{\top}$ is $\left(\sum_{\alpha}|(X - \vp_i)_{\alpha}|\right)^2 \le \left(\sum_{\alpha}X_{\alpha} + \sum_{\alpha}(\vp_i)_{\alpha}\right)^2 \le 4$. So for any $j\in S_g$ and $\vec{M}\in\mathcal{N}$, \begin{equation*}
		\left\langle \left[(X_i - \vp_i)^{\otimes t/2}\right]\left[(X_i - \vp_i)^{\otimes t/2}\right]^{\top}, \vec{M}\right\rangle
	\end{equation*} is a $[-4,4]$-valued random variable, call it $A_i$. By Hoeffding's inequality, \begin{equation}
		\Pr\left[\left|\frac{1}{m}\sum_{i\in S_g}A_i - \frac{1}{m}\sum_{i\in S_g}\E[A_i]\right| \ge \xi\right] \le 2e^{-8m\xi^2},
	\end{equation} so we are done by a union bound over the $M$ directions in $\mathcal{N}$
\end{proof}

We may now proceed with the proof of \eqref{eq:relaxedmoment3}. For $\vec{M}\in\mathcal{K}$, by Lemmas~\ref{lem:contain} and \ref{lem:sumnet}, there is some $\tilde{\vec{M}} = \sum_{\nu}\alpha_{\nu}\vec{M}^*_{\nu}$ such that $\vec{M}^*_{\nu}\in\mathcal{N}$ and $\norm{\vec{M} - \tilde{\vec{M}}}\le\eta$. We may write\begin{align*}\langle \vec{Z}, \vec{M}\rangle &\le \langle \vec{Z},\tilde{\vec{M}}\rangle + \norm{\vec{Z}}_F\norm{\vec{M} - \tilde{\vec{M}}}_F \\ 
&\le \langle \vec{Z},\tilde{\vec{M}}\rangle + \eta\cdot\norm{\vec{Z}}_F \\
&= \sum_{\nu} \alpha_{\nu}\langle \vec{Z},\vec{M}^*_{\nu}\rangle + \eta\cdot\norm{\vec{Z}}_F \\
&\le \sum_{\nu}\alpha_{\nu}\cdot \xi \cdot \norm{\vec{M}^*_{\nu}}_{\max} + \eta\cdot\norm{\vec{Z}}_F \\
&\le 2\xi(\ell\log n + 1)^t + \eta\cdot\norm{\vec{Z}}_F.
\end{align*} where the penultimate step holds with probability $2|\mathcal{N}|e^{-8m\xi^2}$. But observe that because $\norm{\vp_i}_{\infty}\le 1$ for all $i\in[N]$, we have the simple bound that for any $X\in\Delta^n$ and any $i\in[n]$, \begin{align*}
	\norm{\left[(X - \vp_i)^{\otimes t/2}\right]\left[(X - \vp_i)^{\otimes t/2}\right]^{\top}}_F &= \sum_{\alpha,\beta: |\alpha| = |\beta| = t/2}(X - \vp_i)^2_{\alpha}(X - \vp_i)^2_{\beta} \\
	&= \left(\sum_{\alpha}(X^2 + p_i^2 - 2X p_i)_{\alpha}\right)^2 \\
	&\le \left(\sum_{\alpha,\beta}X_{\alpha} + \sum_{\alpha}(p_i)_{\alpha}\right)^2 \le 4,
\end{align*} from which we conclude by triangle inequality that $\norm{\vec{Z}}_F \le 8$.

We conclude that $\langle \vec{Z},\vec{M}\rangle \le 2\xi(\ell\log n + 1)^t\cdot (\log n + 1)^t + 4\eta$, so if $\eta = \frac{1}{8}(8t/k)^{t/2}$ and $\xi = \frac{(8t/k)^{t/2}}{4(\ell\log n + 1)^t}$, then as long as \begin{equation}
	m = \Omega(\xi^{-2}\log|\mathcal{N}|) = \Omega\left(\ell^4 \log^4 n\right)^t \cdot \frac{k^t}{t^{t-1}}\cdot \log(nk/t),
\end{equation} then with probability at least $1 - \poly(n)$, there exists an SoS proof of \eqref{eq:relaxedmoment3} using Axioms~\ref{axioms:momentboundprime}.
\end{proof}

\end{document}